\documentclass[11pt]{article}

\usepackage[utf8]{inputenc}
\usepackage{amsthm,amsmath,amssymb}
\usepackage[margin=.7in]{geometry}
\usepackage{setspace}
\usepackage{graphicx}
\usepackage{subcaption,placeins}
\usepackage{xcolor}
\usepackage{hyperref}
\usepackage{bm}
\usepackage[normalem]{ulem}
\hypersetup{
    colorlinks=true,
    urlcolor=blue,
    citecolor=blue,
    linkcolor=blue
}
\PassOptionsToPackage{hyphens}{url}\usepackage{hyperref}
\PassOptionsToPackage{sloppy}{url}\usepackage{hyperref}
\usepackage{cleveref}
\usepackage{ctable}
\usepackage{enumerate}
\usepackage{indentfirst}
\usepackage{comment}
\usepackage{soul}
\usepackage{ragged2e}

\newcommand{\E}{\mathbb{E}}
\renewcommand{\P}{\textrm{P}}
\newcommand{\indicator}[1]{\mathbf{1}\{#1\}}
\newcommand{\N}{\mathcal{N}}
\newcommand{\bmd}{\bm{\Delta}}

\usepackage[bibstyle=numeric, citestyle=authoryear, doi=false, url=true, backend=biber, maxbibnames=10, maxcitenames=4, uniquelist=false, uniquename=false, sorting=nyt]{biblatex}
\renewbibmacro{in:}{}
\DeclareNameAlias{default}{family-given/given-family}

\newcommand{\citet}{\textcite}
\newtheorem{assumption}{Assumption}
\newtheorem{proposition}{Proposition}
\newtheorem{lemma}{Lemma}
\newtheorem{theorem}{Theorem}
\newtheorem{corollary}{Corollary}
\newtheorem{inneruassumption}{Assumption}
\newenvironment{namedassumption}[1]
  {\renewcommand\theinneruassumption{#1}\inneruassumption}
  {\endinneruassumption}
  
\crefname{assumption}{Assumption}{Assumptions}
\crefname{lemma}{Lemma}{Lemmas}
\crefname{inneruassumption}{Assumption}{Assumptions}

\theoremstyle{definition}
\newtheorem{remark}{Remark}

    \def\independenT#1#2{\mathrel{\setbox0\hbox{$#1#2$}%
    \copy0\kern-\wd0\mkern4mu\box0}}

\title{\textbf{Treatment Effects in Staggered Adoption Designs with Non-Parallel Trends\footnote{Code for the approach proposed in the current paper is available in the \texttt{staggered\_ife2} function from the \texttt{ife R} package which can be downloaded from \url{https://github.com/bcallaway11/ife}.}}}

\author{Brantly Callaway\footnote{Department of Economics, University of Georgia.  Email: \href{mailto:brantly.callaway@uga.edu}{brantly.callaway@uga.edu}} \and Emmanuel Selorm Tsyawo\footnote{FGSES, Universit\'e Mohammed VI Polytechnique. Email: \href{mailto:emmanuel.tsyawo@um6p.ma}{emmanuel.tsyawo@um6p.ma}}}

\begin{document}

\maketitle

\abstract{\noindent This paper considers identifying and estimating causal effect parameters in a staggered treatment adoption setting --- that is, where a researcher has access to panel data and treatment timing varies across units.  We consider the case where untreated potential outcomes may follow non-parallel trends over time across groups.  This implies that the identifying assumptions of leading approaches such as difference-in-differences do not hold.  We mainly focus on the case where untreated potential outcomes are generated by an interactive fixed effects model and show that variation in treatment timing provides additional moment conditions that can be used to recover a large class of target causal effect parameters.  Our approach exploits the variation in treatment timing without requiring either (i) a large number of time periods or (ii) requiring any extra exclusion restrictions. This is in contrast to essentially all of the literature on interactive fixed effects models which requires at least one of these extra conditions.  Rather, our approach directly applies in settings where there is variation in treatment timing.  Although our main focus is on a model with interactive fixed effects, our idea of using variation in treatment timing to recover causal effect parameters is quite general and could be adapted to other settings with non-parallel trends across groups such as dynamic panel data models.}

\bigskip

\bigskip

\bigskip

\bigskip

\bigskip


\bigskip


\vspace{200pt}

\normalsize

\onehalfspacing

\pagebreak

\section{Introduction}

Exploiting access to panel data is one of the most common, if not \textit{the} most common, approach for researchers hoping to learn about the causal effect of economic policies (or some other treatment) on some outcome of interest.  In a setting with panel data, economists have traditionally been most interested in models that include time-invariant unobserved heterogeneity that may be correlated with the treatment variable rather than, say, models that are primarily driven by lagged outcomes.  A main reason for this is that economic theory often suggests models that include variables that may not be observed by the researcher.  Following a large, recent literature on panel data approaches to causal inference that are robust to treatment effect heterogeneity, our starting point is a model for untreated potential outcomes such as
\begin{align} \label{eqn:untreated-potential-outcomes-general}
    Y_{it}(0) = h_t(\xi_i, e_{it})
\end{align}
where $Y_{it}(0)$ is unit $i$'s untreated potential outcome in time period $t$ (i.e., the outcome that unit $i$ would experience in period $t$ if it did not participate in the treatment), $h_t$ is some (unknown) non-parametric function that can change over time, $\xi_i$ is unit-specific time-invariant unobserved heterogeneity (that may be distributed differently for the treated group relative to the untreated group and is not necessarily restricted to be scalar), and $e_{it}$ are time-varying unobservables; for simplicity, we abstract from observed covariates.  This is obviously a challenging model to make much progress with (even under additional conditions on the time-varying unobservables such as being independent of the treatment and $\xi_i$).  Moreover, in a large fraction of applications in economics, researchers only have access to a few periods of data with which to estimate a model or recover some target parameters --- this leads to substantial drawbacks for estimation strategies that rely on estimating $\xi_i$ for each unit due to the incidental parameters problem.  Thus, it is very common to significantly simplify this model to the following 
\begin{align} \label{eqn:untreated-potential-outcomes-twfe}
    Y_{it}(0) = \theta_t + \xi_i + e_{it}
\end{align}
In this setting, if the distribution of $e_{it}$ is the same across groups, then it is straightforward to difference out the unit fixed effects, $\xi_i$, recover $\theta_t$ (given a large number of cross-sectional units), and hence to recover average treatment effect parameters.  In fact, this is exactly the sort of model that leads to difference-in-differences identification strategies which are the dominant approach to causal inference with short panels in economics.\footnote{To be clear, this is not the only viable approach here.  See, for example, \citet{athey-imbens-2006,chernozhukov-val-hahn-newey-2013} for alternative approaches that can work with short panels though, like difference-in-differences, these approaches require additional auxiliary assumptions relative to the model in \Cref{eqn:untreated-potential-outcomes-general}.}

However, we emphasize that, even if many theoretical arguments lead to models such as the one in \Cref{eqn:untreated-potential-outcomes-general}, the extra linearity condition that leads to \Cref{eqn:untreated-potential-outcomes-twfe} is often not implied by economic theory --- despite it being a key requirement for the identification strategy to work.  This is an important distinction relative to using a linear projection to estimate a possibly nonlinear conditional expectation where the conditioning variables are observed.  In this case, the linear projection model has certain good properties (such as being the best linear approximation to conditional expectation) but, unlike \Cref{eqn:untreated-potential-outcomes-twfe}, linearity does not serve an important role as an identification assumption. 

In the current paper, we instead consider the following model for untreated potential outcomes
\begin{align} \label{eqn:untreated-potential-outcomes-ife}
    Y_{it}(0) = \theta_t + \eta_i + \lambda_i'F_t + e_{it}
\end{align}
where we have split the time-invariant unobserved heterogeneity, $\xi_i$, into two components, $\eta_i$ and $\lambda_i$, and allow for the effect of some of the components of the unobserved heterogeneity to vary over time.  This is an interactive fixed-effects model for untreated potential outcomes. Viewed together, $\lambda_i'F_t$ can vary across units and periods, thus %
notably weakening the parallel trends assumption that results from the simpler model in \Cref{eqn:untreated-potential-outcomes-twfe}.  %
Moreover, there are a number of features of this model that are attractive in the context of policy evaluation. For example, \citet{callaway-karami-2023} argue that this model (rather than, say, a unit-specific linear trends model) arises naturally in a setting where a researcher believes that parallel trends holds after conditioning on some other variables, but those variables are not observed by the researcher.\footnote{To give a more specific example, suppose that a researcher was interested in the effect of some treatment on a person's income.  Further, suppose that the researcher thinks that parallel trends holds after conditioning on a person's ability.  If that researcher were using data from the National Longitudinal Survey, then there are measures of a person's ability (such as the person's score on the Armed Forces Qualification Test (AFQT)) that the researcher could use.  On the other hand, if the researcher were using data from the Current Population Survey, there are no obvious measures of ability that are available.  Thus, in the first case, the researcher could just directly include AFQT score in the first case in place of $\lambda_i$ in \Cref{eqn:untreated-potential-outcomes-ife} and rely on a version of conditional parallel trends similar to the case considered in \citet{heckman-ichimura-smith-todd-1998}; in the second case, however, $\lambda_i$ would be unobserved and that would lead to the type of interactive fixed effects model that we consider in the current paper.  Notice that, by contrast, this discussion does not lead to linear trends models (where $\lambda_i'F_t$ is replaced by $\lambda_i t$) as, if ability were observed, it seems very unlikely that any researcher would include it in the model attached to the linear trend term $t$.}

In the current paper, we propose an identification strategy to recover causal effect parameters when (i) untreated potential outcomes are generated by an interactive fixed effects model as in \Cref{eqn:untreated-potential-outcomes-ife} and (ii) in a setting where there is variation in treatment timing across units (often referred to as \textit{staggered treatment adoption}).  Several other papers (e.g., \citet{gobillon-magnac-2016,xu-2017,callaway-karami-2023,imbens-kallus-mao-2021,brown-butts-2022}; see below for additional related discussion) have proposed approaches for recovering treatment effect parameters under an interactive fixed effects model for untreated potential outcomes.  These papers typically either require (i) the number of time periods to be large (in the sense of growing with the sample size rather than being fixed) or (ii) require additional auxiliary conditions such as exclusion restrictions or assumptions about the time-varying error terms being serially uncorrelated.  These types of extra conditions are, however, implausible in many applications in economics, and, moreover, are not required by more common approaches such as difference-in-differences or including unit-specific linear trends.  As discussed earlier, long panels are simply not available in a large fraction of applications in economics (and, even when a long panel is available, it may not be credible that a particular model is stable over a long time horizon; see \citet{arkhangelsky-etal-2021} for more discussion of this point).  Likewise, assumptions that rule out serial correlation in time-varying error terms are often seen to be implausible in applications in economics (for example, this is a main issue discussed in \citet{bertrand-duflo-mullainathan-2004}).  Arguably, these issues at least partially explain the relative lack of popularity of interactive fixed effects models in empirical work in microeconomics.   

Regarding staggered treatment adoption, a number of recent papers have considered difference-in-differences approaches to identifying causal effect parameters in the presence of treatment effect heterogeneity and staggered treatment adoption.  This literature has pointed out a number of weaknesses of two-way fixed effects regressions (which have been the dominant approach to implementing difference-in-differences identification strategies for many years) in this context (\citet{chaisemartin-dhaultfoeuille-2020,borusyak-jaravel-spiess-2022,goodman-2021,sun-abraham-2021}) and proposed alternative estimators that circumvent these issues (\citet{callaway-santanna-2021,gardner-2022}, among others).  These papers all treat variation in treatment timing as a nuisance and propose approaches that sidestep issues that are caused by staggered treatment adoption but do not show up in cases where the treatment timing is common across all units. 

In this paper, we consider a similar setting --- we consider a case with the same data availability and are interested in the same types of causal effect parameters.  However, rather than viewing staggered treatment adoption as a nuisance, we show how staggered treatment adoption can be exploited to recover causal effect parameters under substantially more sophisticated models than are typically used in the difference-in-differences literature.  And, in particular, we show that staggered treatment adoption provides the opportunity to recover causal effect parameters under substantial/complex violations of parallel trends.  Our approach does not require a large number of time periods nor does our approach require any of the auxiliary assumptions mentioned above.  Conceptually, our insight is that, rather than relying on exclusion restrictions or assumptions ruling out serially correlation in the time-varying error terms to introduce additional moment conditions to identify the relevant parameters in a model with a fixed number of time periods, the groups (defined by the timing of treatment) provide an additional source of moment conditions that can be used to identify the relevant parameters in the model.  We also emphasize that, although we mainly focus on an interactive fixed effects model for untreated potential outcomes, this same insight could be quite useful in a number of other applications that require extra moment conditions to identify the model.

\subsubsection*{Related Work}

Our paper is related to a large literature on interactive fixed effects.  Foundational work in this literature includes \citet{pesaran-2006,bai-2009,ahn-lee-schmidt-2013}, among many others.  This literature is vast, so we will mostly confine this section to papers at the intersection of the interactive fixed effects literature and the causal inference literature.  That being said, we note that, of the papers mentioned above, our approach is most closely related to the strand of the literature that builds on \citet{ahn-lee-schmidt-2013} as our approach does not require a large number of time periods and relies on GMM types of identification/estimation arguments.

There are a few recent papers that consider a treatment effects setting in the context of an interactive fixed effects model for untreated potential outcomes.  Some of these papers require access to a large number of periods.  This includes \citet{gobillon-magnac-2016,xu-2017,chan-kwok-2021}.  Work on the synthetic control method also broadly falls in this category; for example, \citet{abadie-diamond-hainmueller-2010} emphasize the connection between the synthetic control method and interactive fixed effects models, and, like the papers mentioned above, the synthetic control method is often rationalized in a setting where the number of time periods is large.   Other related work includes \citet{athey-bayati-doudchenko-imbens-khosravi-2021,bai-ng-2021}; see \citet{liu-wang-xu-2021} for more details on these types of approaches.  Several recent papers have considered interactive fixed effects models for untreated potential outcomes without relying on arguments where the number of time periods needs to be large (in particular, \citet{callaway-karami-2023,imbens-kallus-mao-2021,brown-butts-2022,brown-butts-westerlund-2023}).  This strand of the literature is most similar to what we consider in the current paper.  Unlike these papers, however, we do not require assumptions about time invariance of parameters in the model for untreated potential outcomes (which can allow for particular covariates to be used as instruments) or independence assumptions on the time-varying unobservables (which can allow for outcomes in other time periods to be used as instruments) or other auxiliary conditions.  Instead, we are able to generate moment conditions to identify the parameters in the model for untreated potential outcomes from the staggered nature of the treatment adoption.

Our work also builds on recent papers in the difference-in-differences literature (\citet{chaisemartin-dhaultfoeuille-2020,goodman-2021,callaway-santanna-2021,sun-abraham-2021,marcus-santanna-2021,wooldridge-2021,gardner-2022,borusyak-jaravel-spiess-2022,dube-girardi-jorda-taylor-2023}, among others).  Staggered treatment adoption has been a main case that has been considered in this literature.  However, that literature has seen staggered treatment adoption as a nuisance and proposed ways to circumvent issues with traditional panel data estimation strategies that arise due to staggered treatment adoption.  Notable among these papers, \citet{marcus-santanna-2021} point out that, in a staggered treatment adoption setting, identification strategies that are based on parallel trends assumptions can be highly over-identified.  Given over-identification, they propose more efficient estimators and tests for parallel trends. Our arguments are related to the over-identification case they consider but we instead allow for the extra moment conditions to be used to identify causal effect parameters while substantially relaxing the parallel trends assumption.

More generally, violations of the parallel trends assumption for untreated potential outcomes is a primary concern in empirical work that utilizes panel data in order to identify causal effect parameters.  For example, the vast majority of empirical papers using DID-type identification strategies report event studies that usually include pre-treatment estimates of pseudo-treatment effects with the goal of assessing the credibility of the parallel trends assumption (see \citet{freyaldenhoven-hansen-shapiro-2019,roth-2022} for related discussion). %
When parallel trends seems likely to be violated (or as a robustness check), another common strategy is to introduce individual-specific linear trends in untreated potential outcomes (\citet{heckman-hotz-1989,wooldridge-2005,mora-reggio-2019}) which allow for certain (though relatively limited) violations of parallel trends.  The approach proposed in the current paper generalizes these linear trend models.  Another approach is to allow for violations of parallel trends that result in partial identification of treatment effect parameters of interest (e.g., these ideas include things like allowing for (i) ``not-too-big'' violations of linear trends or (ii) that violations of parallel trends in post-treatment periods are ``not-too-different'' from violations of parallel trends in pre-treatment periods (\citet{manski-pepper-2018,rambachan-roth-2020}).  %

\section{Identification}

\paragraph{Notation}

We consider a case where there are $\mathcal{T}$ time periods of panel data with $n$ units.  We denote a particular time period by $t \in \{1, 2 \ldots, \mathcal{T}\}$.  We focus on the case where $\mathcal{T}$ is fixed.  We consider the case with a binary treatment $D_{it}$ that is equal to 1 if unit $i$ is treated in time period $t$ and is equal to 0 otherwise.  To formalize the idea of staggered treatment adoption, we make the following assumption.
\begin{assumption}[Staggered Treatment Adoption] \label{ass:staggered} For all units and for all time periods $t=2,\ldots,\mathcal{T}$, $D_{it-1}=1 \implies D_{it} = 1$.
\end{assumption}
\Cref{ass:staggered} implies that, once a unit becomes treated, it remains treated in subsequent periods.  Staggered treatment adoption is common in many applications in economics.  For example, many location-specific (e.g., state-level) policies follow a staggered adoption pattern where policies are implemented in different locations at different points in time while remaining in place in subsequent periods (at least over the relatively short time horizons that are the relevant case for the setting that we consider).  Other treatments in economics can be ``scarring'' in the sense that once a unit becomes treated, the unit ``changes state'' and is considered to be treated in subsequent periods as well.  To give a specific example, \citet{sun-abraham-2021} (in the context of discussing \citet{dobkin-finkelstein-kluender-notowidigdo-2018}) consider the effect of hospitalization (the treatment) on various individual-level economic variables.  In this setting hospitalization is scarring in the sense that treated units are not hospitalized year after year, but rather once an individual becomes hospitalized, that individual permanently moves over to being in the treated group.  See also \citet{chaisemartin-dhaultfoeuille-2022a,callaway-2023} for additional discussion about staggered treatment adoption.  

One noteworthy implication of staggered treatment adoption is that we can fully account for a unit's entire treatment history by that unit's ``group''.  We define a unit's group by the time period when the unit becomes treated, and use the notation $G_i$ for this variable.  For units that do not participate in the treatment, we (somewhat arbitrarily) set $G_i = \infty$.  If there are units that are already treated in the first period, we drop those units.\footnote{It is without loss of generality that we have access to a never-treated group.  In applications where all units eventually become treated, it is not possible to use an identification strategy that relies on having a comparison group (such as difference-in-differences or the strategy that we propose in the current paper) to recover treatment effect parameters in periods after all units become treated.  Therefore, our approach in this case (and the approach taken in most other papers with the same sort of data) would be to drop all periods after all units have become treated.  In this case, in the last remaining period, there would be one not-yet-treated group which would be the never-treated group that we refer to here.}  This is standard in the literature on policy evaluation with panel data; for example, difference-in-differences types of identification arguments drop units treated in the first period because (i) they are not useful for learning about the path of untreated potential outcomes (except under strong additional assumptions) and (ii) it is not possible to recover treatment effect parameters for this group as we never observe their untreated potential outcomes --- the same issues apply to the setting that we consider below.  We denote the full set of groups by $\mathcal{G} \subseteq \{2,\ldots,\mathcal{T},\infty\}$.  We additionally denote the set of groups that ever participate in the treatment by $\bar{\mathcal{G}} = \mathcal{G} \setminus \infty$.  To deal with the interactive fixed effects, we also have to drop some early-treated groups, but this depends on the number of interactive fixed effects in the model (in particular, we need to have access to at least $R+1$ pre-treatment periods where $R$ is the number of interactive fixed effects in the model).  We return to this issue below.

Next, we define potential outcomes.  Let $Y_{it}(g)$ denote the outcome that unit $i$ would experience in time period $t$ if it were in group $g$, and, for notational convenience, we define $Y_{it}(0)$ as the outcome that unit $i$ would experience in time period $t$ if it did not participate in the treatment in any time period --- we refer to this as a unit's untreated potential outcome.  In each time period, the observed outcome is given by the potential outcome corresponding to a unit's actual group.  That is, we observe $Y_{it} = Y_{it}(G_i)$.  We make the following assumption.
\begin{assumption}[No Anticipation] \label{ass:no-anticipation}
    For all units and for any time period $t < G_i$ (i.e., pre-treatment time periods for unit $i$), $Y_{it} = Y_{it}(0)$.
\end{assumption}
\Cref{ass:no-anticipation} says that outcomes in pre-treatment periods are not affected by participating in the treatment in subsequent periods.  This assumption is common in the literature though we note that it is straightforward to weaken this assumption to ``limited anticipation'' where outcomes in periods that are ``far enough'' away from the treatment period are not affected by eventually participating in the treatment.  In the current paper, in order to focus on main ideas, we do not consider this extension, but the related ideas in \citet{callaway-santanna-2021,sun-abraham-2021,callaway-karami-2023} would immediately apply to our setting (in particular, these arguments would essentially suggest just to ``back up'' the identification strategy into earlier periods).  Finally, for this section, we make an assumption about the sampling process.
\begin{assumption}[Observed data]\label{ass:sampling}
    The observed data consists of $\{Y_{i1},\dots,Y_{i\mathcal{T}},G_i\}_{i=1}^n$ which are $iid$ where $n$ is the number of units.
\end{assumption}
\Cref{ass:sampling} says that we have access to an iid sample across units.  In practice, this allows, for example, for the outcomes to be serially correlated.  Our identification arguments below apply immediately in settings with clustering, and it is straightforward to extend our inference results to cases with clustering (we provide more details below).

\subsection{Parameters of Interest}

Next, we introduce the parameters that our approach will target below.  Our immediate target parameter of interest is the group-time average treatment effect, which, for $t \geq g$ (post-treatment time periods), is defined as
\begin{align*}
    ATT(g,t) := \E[Y_t(g) - Y_t(0) | G=g].
\end{align*} 
This is the mean difference between treated potential outcomes and untreated outcomes for group $g$ in time period $t$.  Group-time average treatment effects have been emphasized in recent work on difference-in-differences and show up as building blocks both for understanding limitations of TWFE regressions (\citet{chaisemartin-dhaultfoeuille-2020}) and for proposing alternative estimation strategies that circumvent the limitations of TWFE regressions (\citet{callaway-santanna-2021}).  It is important to note that, conditional on $G=g$, $Y_t(g)$ is an observed outcome.  However, $Y_t(0)$ is not an observed outcome.  Thus, the challenge for identifying $ATT(g,t)$ is in recovering $\E[Y_t(0) | G=g]$.

In the treatment effects literature with panel data, it is common to aggregate $ATT(g,t)$'s into lower-dimensional treatment effect parameters.  We briefly discuss the two most popular aggregations, but note that others are also possible (see \citet{callaway-santanna-2021} for alternative aggregations).  We start with an \textit{event-study} type of aggregation.  First, define $e:=t-g$ which denotes the length of exposure to the treatment; for example, $e=0$ when $t=g$ which is the period that units in group $g$ become treated.  Also, define $\mathcal{G}_e := \{ g \in \mathcal{G} | g+e \leq \mathcal{T} \}$, which is the set of groups that are observed to have participated in the treatment for $e$ periods. Then, consider the parameter
\begin{align*}
    ATT^{ES}(e) := \sum_{g \in \mathcal{G}_e}  ATT(g,g+e) \P(G=g|G \in \mathcal{G}_e)
\end{align*}
which is the average effect of participating in the treatment across units that have been exposed to the treatment for exactly $e$ time periods (note that the time period where exposure to the treatment is equal to $e$ can vary across units).

Next, we consider an aggregation into an overall treatment effect  parameter.  As a step in this direction, first define 
\begin{align*}
    ATT^G(g) := \frac{1}{\mathcal{T} - g + 1} \sum_{t=g}^{\mathcal{T}} ATT(g,t)
\end{align*}
which is the average effect of participating in the treatment that units in group $g$ experienced across all their post-treatment time periods.  Then, a natural overall treatment effect parameter is 
\begin{align*}
    ATT^O := \sum_{g \in \bar{\mathcal{G}}} ATT^G(g) \P(G=g|G \in \bar{\mathcal{G}})
\end{align*}
which is the average effect of participating in the treatment across all units that participated in the treatment in any time period.  

Both the event study and overall average treatment effect are commonly reported in applications.  For the identification results below, it is important to note that both of these parameters are weighted averages of $ATT(g,t)$'s (with weights that are straightforward to estimate).  This implies that, if we can identify each $ATT(g,t)$, then it will be possible to translate $ATT(g,t)$'s into event study parameters or an overall treatment effect parameter if these are the ultimate target parameter(s) for a particular application.  Thus, our arguments below focus on identifying group-time average treatment effects.  Finally, some of our identification arguments result in the identification of group-time average treatment effects for a subset of groups or time periods; in those cases, the aggregated parameters that we discuss here may only be identified for a restricted set of groups and time periods as well.  We defer these sorts of issues to later in the paper.

\subsection{Identifying Group-Time Average Treatment Effects}

Given the framework discussed above, we now introduce an interactive fixed effects model for untreated potential outcomes and, subsequently, our approach to identifying group-time average treatment effects in this context.  We make the following assumptions:
\begin{assumption}[Interactive Fixed Effects Model for Untreated Potential Outcomes] \label{ass:ife}
\begin{equation}
    \label{eqn:IFEmodel}
    Y_{it}(0) = \theta_t + \eta_i + \lambda_i' F_t + e_{it}
\end{equation}
where $\lambda_i$ and $F_t$ are $R$ dimensional vectors.
\end{assumption}
\begin{assumption}[Unconfoundedness Conditional on Unobserved Heterogeneity] \label{ass:sel}
    \begin{align*}
    \E[Y_{t}(0) |\eta, \lambda, G] = \E[Y_{t}(0) |\eta, \lambda] \quad  \textrm{a.s.}
    \end{align*}
\end{assumption}

\Cref{ass:ife} says that untreated potential outcomes are generated by an interactive fixed effects model.  Because we consider a setting with a fixed number of time periods, we treat $\theta_t$ and $F_t$ as being fixed parameters (or, alternatively, our approach can be seen as conditional on the realizations of $\theta_t$ and $F_t$).  In general, because the number of cross-sectional units is large, our strategy will be to consistently estimate (functions of) these parameters.  On the other hand, we treat $\eta_i$ and $\lambda_i$ as being random.  We also interpret $\eta_i$ and $\lambda_i$ as unobserved heterogeneity that can be distributed differently across groups.  These differences in distribution can lead to different levels and trends of untreated potential outcomes for different groups.  The literature on interactive fixed effects models often refers to $F_t$ as \textit{factors} and $\lambda_i$ as \textit{factor loadings}.  For convenience, we sometimes use this terminology below, but for the most relevant applications to our approach (ones with a binary treatment and  fixed-$\mathcal{T}$), interpreting $\lambda_i$ as unobserved heterogeneity and $F_t$ as a time-varying effect of unobserved heterogeneity is probably most natural.   

The model in \Cref{ass:ife} reduces to the sort of TWFE model for untreated potential outcomes that leads to DID identification strategies (see, e.g., \citet{blundell-dias-2009}) when either $F_t$ is constant across $t$ (in this case the interactive fixed effects term is absorbed into the individual fixed effect $\eta_i$) or if $\lambda$ has the same mean across groups.\footnote{Related to this discussion, we also explicitly include unit and time fixed effects.  Earlier work typically noted that two-way fixed effects models were special cases of interactive fixed effects models, but it is common in more recent work to explicitly (and separately) include the two-way structure (see, for example, \citet{callaway-karami-2023,brown-butts-2022} for related discussion).} A related side-effect of including explicit unit and time fixed effects is that the dimension of the interactive fixed effect term, is governed by the number of factors that vary over time and by the number of factor loadings whose means vary across groups (see the discussion below for more details).  

\Cref{ass:sel} says that, if one could observe/condition on the unobserved heterogeneity terms $\eta$ and $\lambda$ then the average untreated potential outcome would be the same across groups.  Another way to think about this assumption is that, in terms of generating untreated potential outcomes, the important differences between groups are due to their distribution of $\eta$ and $\lambda$.  This sort of assumption is very common in the literature on treatment effects with panel data; see, for example, \citet{gobillon-magnac-2016,xu-2017,gardner-2020,callaway-karami-2023}. Importantly, these assumptions do not put any structure on how treated potential outcomes are generated.  They also allow for units to select into participating in the treatment on the basis of their treated potential outcomes and their unobserved heterogeneity ($\eta$ and $\lambda$).

An implication of \Cref{ass:ife,ass:sel} is that 
\begin{align} \label{eqn:Eu}
    \E[e_t |\eta, \lambda, G] = 0
\end{align}
which we use below as a source of moment conditions to identify parameters from the interactive fixed effects model. 

\Cref{ass:staggered,ass:no-anticipation,ass:sampling,ass:ife,ass:sel} are the main assumptions that we make in the paper (up to some rank conditions discussed in the next sections).  Before continuing, it is worth emphasizing what we have \textit{not assumed}.  First, 
notice from the conditional exogeneity condition in \Cref{eqn:Eu} that correlation in $e_{it}$ across $t$ is not restricted.  This rules out strategies that rely on using outcomes in other periods as an additional source of identifying information as in some of the arguments in \citet{callaway-karami-2023} and \citet{imbens-kallus-mao-2021}.  Serial correlation in $e_{it}$ is generally thought to be quite prevalent in most of the fixed-$\mathcal{T}$ policy evaluation settings that our approach is relevant to (see, in particular, \citet{bertrand-duflo-mullainathan-2004}).  Second, we do not require any extra conditions on covariates that enter the model as in \citet{callaway-karami-2023,brown-butts-2022,brown-butts-westerlund-2023} that allow them to be used as excluded instruments.\footnote{In our case, there are no covariates that are even included in the model.  To be clear, it would be straightforward to include covariates in the approach that we propose below.  These could be handled in analogous ways to including covariates in typical panel data applications (and are, therefore, in some sense, not very interesting from an econometrics standpoint).  %
This is in contrast to the other approaches mentioned above that exploit restrictions on the covariates (such as restrictions on how the covariates effects can vary over time or by imposing auxiliary models for the covariates) to achieve identification.  See \Cref{rem:covariates} below for additional discussion along these lines.}   In contrast to these approaches, below we will exploit the staggered treatment adoption in order to achieve identification.

\subsubsection{Baseline Case with One Interactive Fixed Effect}\label{SubSect:Baseline}  To start with, we consider a small case that is helpful to understand the identification strategy in the current paper.  In the next section, we consider generalizations allowing for (i) more time periods, (ii) more interactive fixed effects, and (iii) more groups. For now, suppose that $R=1$, so that the model in \Cref{ass:ife} becomes
\begin{align} \label{eqn:ife-reg-baseline}
Y_{it}(0) = \theta_t + \eta_i + \lambda_i F_t + e_{it}
\end{align}
In addition, suppose that $\mathcal{T}=4$ and $\mathcal{G} = \{3,4,\infty\}$ so that there is a group that is treated in periods 3, 4, and an untreated group.  Here, we focus on identifying $ATT(3,3)$ (the average effect of participating in the treatment for group 3 in period 3).  First, notice that
\begin{align} \label{eqn:att-3-3}
    ATT(3,3) &= \E[Y_3(3) - Y_3(0) | G=3] \nonumber \\
    &= \E[Y_3(3) - Y_2(0) | G=3] - \E[Y_3(0) - Y_2(0) | G=3] \nonumber \\
    &= \E[\Delta Y_3 | G=3] - \E[\Delta Y_3 (0)| G=3]
\end{align}
where we use the notation $\Delta Y_t := Y_t - Y_{t-1}$ and where the first equality is just the definition of $ATT(3,3)$, the second equality adds and subtracts $\E[Y_2(0)|G=3]$, which is the mean untreated potential outcome for group 3 in time period 2, and the last equality holds because $Y_3(3)$ and $Y_2(0)$ are observed outcomes for group 3.  \Cref{eqn:att-3-3} highlights that, in order to recover $ATT(3,3)$, the key identification challenge is to recover $\E[\Delta Y_3(0)|G=3]$, that is, how untreated potential outcomes would have changed over time for group 3 had it not become treated in period 3.  For this section, we make the following additional assumptions.
\begin{namedassumption}{6-Baseline}[Factor Rank Condition] $F_2 \neq F_1$.  \label{ass:factor-baseline}  
\end{namedassumption}
\begin{namedassumption}{7-Baseline}[Factor Loading Rank Condition] $\E[\lambda|G=4] \neq \E[\lambda|G=\infty]$. \label{ass:loading-baseline}
\end{namedassumption}

\noindent \Cref{ass:factor-baseline} says that the factors change between the first two periods.  \Cref{ass:loading-baseline} says that the mean of $\lambda$ is different between group 4 and the never-treated group (which are the two relevant comparison groups in this section).  We discuss both of these assumptions and their practical significance in more detail at the end of this section. %

As a first step towards recovering $\E[\Delta Y_3(0) | G=3]$ from \Cref{eqn:att-3-3}, notice that
\begin{align}
    \Delta Y_{i3}(0) = \Delta \theta_3 + \lambda_i\Delta F_3 + \Delta e_{i3}. \label{eqn:Delta-y3}
\end{align}
Similarly,
\begin{align}
    \Delta Y_{i2}(0) = \Delta \theta_2 + \lambda_i \Delta F_2 + \Delta e_{i2} \label{eqn:Delta-y2}
\end{align}
and this second equation implies that
\begin{align}
    \lambda_i &= \Delta F_2^{-1} \Big( \Delta Y_{i2}(0) - \Delta \theta_2 - \Delta e_{i2} \Big). \label{eqn:lambda}
\end{align}
where this expression holds just by rearranging terms from \Cref{eqn:Delta-y2}.  Notice that this step uses \Cref{ass:factor-baseline} to avoid dividing by 0.  Thus, from plugging the expression for $\lambda_i$ in \Cref{eqn:lambda} back into \Cref{eqn:Delta-y3}, we have that
\begin{align}
    \Delta Y_{i3}(0) &= \Delta \theta_3 + \lambda_i \Delta F_3 + \Delta e_{i3} \nonumber \\
    &= \left( \Delta \theta_3 - \frac{\Delta F_3}{\Delta F_2}\Delta \theta_2\right) + \frac{\Delta F_3}{\Delta F_2}\Delta Y_{i2}(0) + \left( \Delta e_{i3} - \frac{\Delta F_3}{\Delta F_2}\Delta e_{i2} \right) \nonumber \\
    & =: \theta_{3}^* + F_{3}^* \Delta Y_{i2}(0) + v_{i3} \label{eqn:ife-reg} 
\end{align}
where we define $\theta_3^* := \Delta \theta_3 - \frac{\Delta F_3}{\Delta F_2}\Delta \theta_2$, $F_3^* := \frac{\Delta F_3}{\Delta F_2}$, and $v_{i3} := \Delta e_{i3} - \frac{\Delta F_3}{\Delta F_2}\Delta e_{i2}$. Taking the expectation of \Cref{eqn:ife-reg} conditional on being in group 3, we have that
\begin{align} \label{eqn:Delta-y3-g3}
    \E[\Delta Y_3(0) | G=3] &=\theta_{3}^* + F_{3}^* \E[\Delta Y_{2} | G=3].
\end{align}
\Cref{eqn:Delta-y3-g3} holds from \Cref{eqn:ife-reg} because (i) $\Delta Y_{i2}(0)$ is observed for units in group 3 (thus, the untreated potential outcomes on the right hand side of \Cref{eqn:ife-reg} are observed outcomes for group 3), and (ii) $\E[v_{3} | G=3] = \E\Big[ \E[v_3|\eta, \lambda, G=3] \Big| G=3\Big] = 0$ (which holds by the law of iterated expectations and then by \Cref{ass:sel}).  

\Cref{eqn:Delta-y3-g3} therefore suggests that identifying $ATT(3,3)$ hinges on identifying the two-dimensional parameter vector $(\theta_{3}^*, F_{3}^*)'$.  Towards this end, notice that for groups 4 and $\infty$, which are not-yet-treated in period 3, $\Delta Y_{i3}(0)$ is observed. This suggests the possibility of recovering $(\theta_3^*,F_3^*)'$ using data coming from those groups through \Cref{eqn:ife-reg}.  This is the gist of the strategy that we use below.  An important complication arises, however, because, by construction, $\Delta Y_{i2}(0)$ is correlated with $v_{i3}$ in \Cref{eqn:ife-reg} --- this is because $v_{i3}$ contains $\Delta e_{i2}$.  This rules out recovering $(\theta_3^*,F_3^*)'$ directly from the regression of $\Delta Y_{i3}$ on $\Delta Y_{i2}$ using units that are not-yet-treated by period 3. This is the same type of issue that shows up in all of the literature on treatment effects in interactive fixed effects in fixed-$\mathcal{T}$ settings; it is the reason that the existing approaches discussed above impose extra conditions to generate additional moment conditions to be able to recover parameters from the interactive fixed effects model and, subsequently, treatment effect parameters.  

The key insight of our identification strategy is that, if we have two distinct groups that are not-yet-treated in period 3, those two groups provide two moment conditions that can potentially be used to recover the two parameters $\theta_3^*$ and $F_3^*$.  In particular, \Cref{ass:sel} implies that, for any group $g$ and for any time period $t$
\begin{align*}
    \E[e_t | G=g] = \E\Big[ \E[e_t | \eta, \lambda, G=g] | G=g\Big] = 0
\end{align*}
where the first equality holds by the law of iterated expectations and the second equality holds by \Cref{ass:sel}.  In the setting considered here, this further implies that, for any group $g$ 
\begin{align*}
    \E[v_3 | G=g ] &= \E[\Delta e_3 | G=g] - F_3^* \E[\Delta e_2 | G=g] = 0
\end{align*}
This implies the following moment conditions
\begin{align}
    0 = \E[\Delta Y_3(0) | G=g] - \Big(\theta_3^* + F_3^* \E[\Delta Y_2(0) | G=g]\Big) \label{eqn:moment-conditions}
\end{align}
$\Delta Y_{i3}(0)$ is not observed for units in group 3, so this moment condition is infeasible to estimate for group 3 (i.e., as expected, group 3 is not useful for recovering $(\theta_3^*,F_3^*)'$); however, $\Delta Y_{i3}(0)$ is observed for units in group 4 and the never-treated group.  Thus, there are two available moment conditions and two parameters to estimate which implies that the necessary order condition is satisfied here.  From these two moment equations separately for group 4 and the never-treated group (and after some straightforward algebra), it can be shown that
\begin{align} \label{eqn:F3-star-baseline}
    F_3^* &= \frac{\E[\Delta Y_3 | G=\infty] - \E[\Delta Y_3 | G=4]}{\E[\Delta Y_2| G=\infty] - \E[\Delta Y_2 | G=4]} \\[10pt]
    \text{ and}  \qquad \theta_3^* &= \E[\Delta Y_3 | G=4] - F_3^* \E[\Delta Y_2 | G=4]
\end{align}
which implies that $F_3^*$ and $\theta_3^*$ are identified.  %

Before continuing, it is worth briefly commenting on the role that \Cref{ass:factor-baseline,ass:loading-baseline} play in the discussion above. Towards this end, notice that for $t=2,3$, 
\begin{align}
    \E[\Delta Y_t | G=\infty] - \E[\Delta Y_t | G=4] = \Big(\E[\lambda|G=\infty] - \E[\lambda|G=4]\Big) \Delta F_t \label{eqn:F3-explanation}
\end{align}
which follows from the model for untreated potential outcomes in \Cref{eqn:ife-reg-baseline} and similar arguments as are used throughout this section.  Notice that these are the expressions that show up in the numerator and in the denominator for $F_3^*$ above, and further notice that these expressions depend on differences in the mean of $\lambda$ across groups.  %
Next, it is clear that $F_3^*$ will not be identified if $\E[\Delta Y_2 | G=\infty] = \E[\Delta Y_2 | G=4]$ (this is the term that shows up in the denominator of the expression for $F_3^*$ above).  From \Cref{eqn:F3-explanation}, we can see that both assumptions are required for this difference to be non-zero.  More intuitively, by construction our setup rules out (i) $F_1=F_2=F_3=F_4$ (otherwise, the interactive fixed effect term would be absorbed into the unit fixed effect $\eta_i$) and also rules out (ii) $\E[\lambda | G=3] = \E[\lambda|G=4] = \E[\lambda | G=\infty]$ (otherwise, the interactive fixed effect term would be absorbed into the time fixed effect $\theta_t$); in either case, given our setup, it would imply that $R=0$ rather than $R=1$.   \Cref{ass:factor-baseline,ass:loading-baseline} impose additional requirements relative to this baseline.  First, if $F_1=F_2$, then all three groups will have the same trends in outcomes between the first two periods --- and these are the only two pre-treatment periods for group 3.  In this case, we would not be able to distinguish between effects of the treatment in group 3 relative to changes in $F_3$ (i.e., differences in mean outcomes across groups in period 3 could arise for either of those reasons).  \Cref{ass:factor-baseline} rules out this case.  Second, if $\E[\lambda|G=4] = \E[\lambda|G=\infty]$, then group 4 and the never-treated group will have the same trends in outcomes over time, regardless of changes in $F_t$ over time, and, therefore, would not be useful for learning about changes in $F_t$ in periods after group 3 becomes treated.  \Cref{ass:loading-baseline} rules out this case.  

Given that $F_3^*$ and $\theta_3^*$ are identified, it immediately follows that $ATT(3,3)$ is identified and is given by
\begin{align*}
    ATT(3,3) = \E[\Delta Y_3 | G=3] - \Big( \theta_3^* + F_3^* \E[\Delta Y_2 | G=3] \Big)
\end{align*}

To conclude this section, notice that the main cost of our approach is that we are not able to identify as many group-time average treatment effects as would be possible using other identification strategies.  For example, \citet{callaway-karami-2023} suppose that the researcher has access to a time-invariant covariate that does not affect the path of untreated potential outcomes and show that this sort of covariate can be used to generate additional moment conditions to identify the parameters in the model for untreated potential outcomes.  Our approach does not require this sort of covariate (which is a key advantage of our approach), but it comes at the cost of only being able to identify $ATT(3,3)$ without recovering $ATT(3,4)$ or $ATT(4,4)$ which would be feasible using the approach in \citet{callaway-karami-2023}.

\subsubsection{General Case} \label{sec:general-case}

In this section, we extend the results above to a setting with more periods, more groups, and allow for more interactive fixed effects.  The arguments in this section target recovering $ATT(g,t)$ for a particular group $g \in \bar{\mathcal{G}}$.  %
Given the model in \Cref{ass:ife}, after taking first differences, we have that
\begin{align*}
    \Delta Y_{it}(0) = \Delta \theta_t + \lambda_i'\Delta F_t + \Delta e_{it}
\end{align*}
which eliminates the unit fixed effect $\eta_i$. Next, define
\begin{align*}
    \Delta Y_i(0) := \begin{bmatrix}
        \Delta Y_{i2} \\ \vdots \\ \Delta Y_{i\mathcal{T}}
    \end{bmatrix}; \quad \Delta \theta := \begin{bmatrix}
        \Delta \theta_2 \\ \vdots \\ \Delta \theta_{\mathcal{T}}
    \end{bmatrix}; \quad \bmd \mathbf{F} := \begin{bmatrix}
        \Delta F_2' \\ \vdots \\ \Delta F_{\mathcal{T}}'
    \end{bmatrix}; \text{ and } \Delta e_i := \begin{bmatrix}
        \Delta e_{i2} \\ \vdots \\ \Delta e_{i\mathcal{T}}
    \end{bmatrix}
\end{align*}
where $\Delta Y_i(0)$, $\Delta \theta$, and $\Delta e_i$ are all $(\mathcal{T}-1)$ dimensional vectors, and $\bmd \mathbf{F}$ is $(\mathcal{T}-1)\times R$ matrix.  Similarly, define
\begin{align*}
    \Delta Y_i^{pre(g)}(0) := \begin{bmatrix}
        \Delta Y_{i2} \\ \vdots \\ \Delta Y_{ig-1}
    \end{bmatrix}; \quad \Delta \theta^{pre(g)} := \begin{bmatrix}
        \Delta \theta_2 \\ \vdots \\ \Delta \theta_{g-1}
    \end{bmatrix}; \quad \bmd \mathbf{F}^{pre(g)} := \begin{bmatrix}
        \Delta F_2' \\ \vdots \\ \Delta F_{g-1}'
    \end{bmatrix}; \text{ and } \Delta e_i^{pre(g)} := \begin{bmatrix}
        \Delta e_{i2} \\ \vdots \\ \Delta e_{ig-1}
    \end{bmatrix}
\end{align*}
where $\Delta Y_i^{pre(g)}(0)$, $\Delta \theta^{pre(g)}$, and $\Delta e_i^{pre(g)}$ are all $(g-2) \times 1$ vectors and $\bmd \mathbf{F}^{pre(g)}$ is a $(g-2)\times R$ matrix.  Next, define $p_g := \P(G=g)$ and 
\begin{align*}
    \mathbf{\Lambda} := \E\left[ \frac{\indicator{G=g}}{p_g} \begin{pmatrix} 1 & \lambda' \end{pmatrix} \right]_{g \in \mathcal{G}}
\end{align*}
where the notation $\Big[\ \cdot \ \Big]_{g \in \mathcal{G}}$ indicates that there is one row in the matrix for each group in $\mathcal{G}$.  Thus, $\mathbf{\Lambda}$ is a $|\mathcal{G}| \times (R+1)$ matrix where $|\mathcal{G}|$ denotes the cardinality of the set $\mathcal{G}$ (which is the number of groups). 
 By construction, we have that $\textrm{Rank}(\bmd \mathbf{F}) = R$ and that $\textrm{Rank}(\mathbf{\Lambda}) = R+1$. For example, if it were the case that $\textrm{Rank}(\bmd \mathbf{F}) = (R-1)$, then the unit fixed effect would effectively absorb one of the interactive fixed effects terms and the model could equivalently be re-written with $R-1$ factors.  Similarly, if $\textrm{Rank}(\mathbf{\Lambda}) = R$, then the time fixed effect would effectively absorb one of the interactive fixed effects terms and the model could equivalently be re-written with $R-1$ factors (see \Cref{app:rank-explanation} for a more detailed explanation).   These conditions also impose that $(\mathcal{T}-1) \geq R$ (a restriction on having enough time periods) and that $|\mathcal{G}| \geq (R+1)$ (a restriction on having enough groups); however, we strengthen both of these conditions in the discussion below.  

Below, we focus on identifying treatment effects for a particular group $g$ in a particular post-treatment period $t \geq g$.  Define $\mathcal{G}^{comp}(g,t) := \{ g' \in \mathcal{G} : t < g' \}$. (i.e., this is the set of groups that have not yet been treated by period $t$). Additionally, define 
\begin{align*}
    \ell^{comp}(g,t) := \begin{bmatrix} \frac{\indicator{G=g'}}{p_{g'}} \end{bmatrix}_{g' \in \mathcal{G}^{comp}(g,t)}
\end{align*}
which is an $|\mathcal{G}^{comp}(g,t)|$ dimensional vector, and also define
\begin{align*}
    \mathbf{\Lambda}^{comp}(g,t) := \E\left[ \ell^{comp}(g,t) \begin{pmatrix} 1 & \lambda'\end{pmatrix} \right]
\end{align*}
which is a $|\mathcal{G}^{comp}(g,t)| \times (R+1)$ matrix.  We make the following two assumptions

\begin{assumption}\label{ass:rank-delta-F} $\textrm{Rank}(\bmd \mathbf{F}^{pre(g)}) = R$
\end{assumption}

\begin{assumption}\label{ass:rank-Lambda} $\textrm{Rank}(\mathbf{\Lambda}^{comp}(g,t)) = R+1$
\end{assumption}
\Cref{ass:rank-delta-F,ass:rank-Lambda} generalize \Cref{ass:factor-baseline,ass:loading-baseline} from the case with a single factor in the previous section to the case with $R$ factors considered here.  The intuition is similar too.  \Cref{ass:rank-delta-F} requires that there be enough variation in the factors themselves in pre-treatment periods (and, like the earlier case, rules out cases where there is limited variation in the factors in early periods but where there is more variation in the factors in post-treatment periods). \Cref{ass:rank-Lambda} requires enough variation in the means of $\lambda$ across the available comparison groups.  Besides these similarities, in this case, these conditions also imply some additional restrictions on the groups and time periods for which our approach can recover $ATT(g,t)$.  First, \Cref{ass:rank-delta-F} immediately requires that $(g-2) \geq R$.  This means that we have enough pre-treatment periods given the number of factors $R$.  And, in particular, this means that we can only recover $ATT(g,t)$ for groups for which we observe at least $R+1$ pre-treatment periods (e.g., if $R=2$, then we can only identify treatment effects for groups that become treated in period 4 or later).   Second, \Cref{ass:rank-delta-F} implies that we must have that $|\mathcal{G}^{comp}(g,t)| \geq (R+1)$.  This limits the periods for which we can recover $ATT(g,t)$.  In particular, define $t^{max}(g)$ to be the largest value of $t$ such that $|\mathcal{G}^{comp}(g,t)| \geq R+1$. 
Then, this condition rules out recovering $ATT(g,t)$ in periods $t > t^{max}(g)$ (i.e., later periods where there is not a large enough set of groups that have not been treated).  By extension, for any group that becomes treated in these late periods, we are not able to recover any group-time average treatment effects for that group.  Given some number of factors $R$, we denote the set of groups that satisfy both criteria mentioned above (i.e., that both have enough pre-treatment periods and that are treated early enough to have a large enough comparison group) as $\mathcal{G}^\dagger$; that is, $\mathcal{G}^\dagger := \{g \in \mathcal{G} : (R+2) \leq g \leq t^{max}(g) \}$. %
 This is the set of groups for which we will recover $ATT(g,t)$ using our approach.

To give an example, suppose that $R=2$ and that $\mathcal{G} = \{2,\ldots,\mathcal{T},\infty\}$ (relative to the notation above, here we are supposing that we have groups that become treated in every period as well as a never-treated group).  In this case, $|\mathcal{G}^{comp}(g,t)| = \mathcal{T} - t + 1$; if $R=2$, then $\mathcal{G}^\dagger = \{4, \ldots , \mathcal{T}-2\}$, and we can recover $ATT(g,t)$ for any group $g \in \mathcal{G}^\dagger$ in periods $g \leq t \leq (\mathcal{T}-2 )$ (i.e., all post-treatment periods for group $g$ except the last two periods: $(\mathcal{T}-1)$ and $\mathcal{T}$).

Next, we provide our main identification arguments.  Given the interactive fixed effects model in \Cref{ass:ife}, notice that
\begin{align}\label{eqn:DelY_pre}
    \Delta Y_i^{pre(g)}(0) = \Delta \theta^{pre(g)} + \bmd \mathbf{F}^{pre(g)}\lambda_i + \Delta e_i^{pre(g)}
\end{align}
which implies that
\begin{align}\label{eqn:lambda_i}
    \lambda_i = \mathbf{H}^{pre(g)} \Big( \Delta Y_i^{pre(g)}(0) - \Delta \theta^{pre(g)} - \Delta e_i^{pre(g)}\Big)
\end{align}
where $\mathbf{H}^{pre(g)} := (\mathbf{\Omega}' \bmd \mathbf{F}^{pre(g)})^{-1} \mathbf{\Omega}'$ and $\mathbf{\Omega}$ is a known (or consistently estimable) $(g-2) \times R$ matrix with rank $R$ that satisfies $\textrm{Rank}(\mathbf{\Omega}'\bmd \mathbf{F}^{pre(g)}) = \textrm{Rank}(\bmd \mathbf{F}^{pre(g)})$.\footnote{In practice, there are several leading candidates for $\mathbf{\Omega}$. First, the approach suggested in \citet{callaway-karami-2023} effectively sets $\mathbf{\Omega} = [\mathbf{0}_{R\times (g-2)-R}, \mathbf{I}_R ]'$.  A second candidate is to use the first $R$ principal components of $\Delta Y_i^{pre(g)}$, so that $\mathbf{\Omega}$ is the probability limit of the first $R$ eigenvectors of $\bmd {\mathbf{Y}^{pre(g)}}^{'} \bmd \mathbf{Y}$ (here $\bmd \mathbf{Y}^{pre(g)}$ is the $n^{comp(g,t)} \times (g-2)$ data matrix of pre-treatment outcomes where $n^{comp(g,t)}$ is the number of units in the comparison group for group $g$ in period $t$).  There are tradeoffs to different choices here.  For example, the first choice mentioned above for $\mathbf{\Omega}$ results in only using the most recent $R$ periods of $\Delta Y_{it}(0)$ on the right-hand side of \Cref{eqn:lambda_i} rather than all pre-treatment periods.  \citet{callaway-karami-2023} argue that using fewer pre-treatment periods has an advantage of being more robust to the model for untreated potential outcomes not holding in periods further away from the treatment.  But it comes at the cost of requiring that $\begin{bmatrix} \Delta F_{g-R} & \cdots & \Delta F_{g-1} \end{bmatrix}'$ has rank $R$ (which strengthens \Cref{ass:rank-delta-F}); note that this approach could be tweaked to ``select'' $\Delta F_t$ in periods where the researcher is confident that the rank condition holds for that subset of periods.  On the other hand, using the principal components uses data from all periods, but $\mathbf{\Omega}$ must be estimated and there may be other auxiliary conditions needed for $\textrm{Rank}\big(\mathbf{\Omega}'\bmd \mathbf{F}^{pre(g)}\big) = R$. The most appropriate choice for $\mathbf{\Omega}$ may vary across different applications.}
\Cref{eqn:lambda_i} holds from \Cref{eqn:DelY_pre} by multiplying by $\mathbf{H}^{pre(g)}$ and then cancelling and re-arranging terms.
Moreover, for some particular post-treatment period $t$, we have that
\begin{align}
    Y_{it}(0) - Y_{ig-1}(0) &= (\theta_t - \theta_{g-1}) + (F_t - F_{g-1})'\lambda_i + (e_{it} - e_{ig-1}) \nonumber \\
    &= (\theta_t - \theta_{g-1}) + (F_t - F_{g-1})'\mathbf{H}^{pre(g)}\Big( \Delta Y_i^{pre(g)}(0) - \Delta \theta^{pre(g)} - \Delta e_i^{pre(g)}\Big) + (e_{it} - e_{ig-1}) \nonumber \\
    &= \theta^*(g,t) +  \widetilde{\Delta Y}_i^{pre(g)}(0)'F^*(g,t) + v_i(g,t) \label{eqn:mp-estimating-equation}
\end{align}
where 
\begin{align*}
    \widetilde{\Delta Y}_i^{pre(g)}(0) &:= \mathbf{\Omega}'\Delta Y_i^{pre(g)} \\
    \theta^*(g,t) &:= (\theta_t - \theta_{g-1}) - (F_t - F_{g-1})' \mathbf{H}^{pre(g)} \Delta \theta^{pre(g)} \\
    F^*(g,t) &:= \Big((F_t-F_{g-1})'(\mathbf{\Omega}'\bmd \mathbf{F}^{pre(g)})^{-1}\Big)' \\
    v_i(g,t) &:= (e_{it} - e_{ig-1}) - (F_t - F_{g-1})' \mathbf{H}^{pre(g)} \Delta e_i^{pre(g)}
\end{align*}

Notice that $\theta^*(g,t)$ is scalar, $F^*(g,t)$ is a $R \times 1$ vector and $v_i(g,t)$ is scalar.  Similar to the simpler case discussed above, we will target estimating $\theta^*(g,t)$ and $F^*(g,t)$ below.  Because $\mathbf{\Omega}$ is known, $\widetilde{\Delta Y}_i^{pre(g)}$ amounts to being an $R$ dimensional vector of regressors in \Cref{eqn:mp-estimating-equation}.  Furthermore, for any group $g' \in \mathcal{G}$, $\E[v_i(g,t)|G=g'] = 0$ (which holds under \Cref{ass:sel}).  Similar to the simpler case discussed above, these will be the source of moment conditions that we use below.  
In order to show that $ATT(g,t)$ is identified, we proceed in two steps.  First, we show that $ATT(g,t)$ is identified if we can identify $\theta^*(g,t)$ and $F^*(g,t)$.  Second, we show that (under certain conditions), $\theta^*(g,t)$ and $F^*(g,t)$ are indeed identified.  Towards showing the first part, notice that
\begin{align}
    ATT(g,t) &= \E[Y_t(g) - Y_t(0) | G=g] \nonumber \\
    &= \E[Y_t(g) - Y_{g-1}(0) | G=g] - \E[Y_t(0) - Y_{g-1}(0) | G=g] \nonumber \\
    &= \E[Y_t - Y_{g-1} | G=g] - \Big( \theta^*(g,t) + F^*(g,t)'\E[\widetilde{\Delta Y}^{pre(g)} | G=g] \Big) \label{eqn:attgt-identified}
\end{align}
where the first equality comes from the definition of $ATT(g,t)$, the second equality adds and subtracts $\E[Y_{g-1}(0)|G=g]$, and the third equality holds by (i) plugging in \Cref{eqn:mp-estimating-equation}, (ii) because $\E[v_i(g,t) | G=g] = 0$, and (iii) by replacing potential outcomes with their observed counterparts.  Given that $\theta^*(g,t)$ and $F^*(g,t)$ are known, this implies that $ATT(g,t)$ is identified and completes the first part of the argument described above.  

Next, we move to showing that $\theta^*(g,t)$ and $F^*(g,t)$ are indeed identified.  For a group $g' \in \mathcal{G}^{comp}(g,t)$, we have that
\begin{align*}
    0 = \E\left[ \frac{\indicator{G=g'}}{p_{g'}} v_i(g,t) \right] = \E\left[ \frac{\indicator{G=g'}}{p_{g'}} \left\{ \Big( Y_{t} - Y_{g-1}\Big) - \Big(\theta^*(g,t) - {\widetilde{\Delta Y}}^{{pre(g)}^{'}} F^*(g,t) \Big) \right\}  \right]
\end{align*}
This gives one moment condition for each group in $\mathcal{G}^{comp}(g,t)$.  Therefore, the order condition is that $|\mathcal{G}^{comp}(g,t)| \geq R+1$.  For example, in the case where $R=1$, there are two parameters to identify: $\theta^*(g,t)$ and $F^*(g,t)$ (which is a scalar in this case).  Thus, we need to have at least two groups that are not-yet-treated by period $t$. Stacking the above moment conditions, we have that 
\begin{align} \label{eqn:mom-conds}
    \mathbf{0}_{|\mathcal{G}^{comp}(g,t)|} = \E\left[ \ell^{comp}(g,t) \left\{ \Big( Y_{t} - Y_{g-1}\Big) - \Big(\theta^*(g,t) - {\widetilde{\Delta Y}}^{{pre(g)}^{'}} F^*(g,t) \Big) \right\}  \right]
\end{align}
Then, identification hinges on the matrix
\begin{align*}
    \mathbf{\Gamma}(g,t) := \E\left[ \ell^{comp}(g,t) \begin{pmatrix} 1 \\ \widetilde{\Delta Y}^{pre(g)} \end{pmatrix}' \right]
\end{align*}
$\mathbf{\Gamma}(g,t)$ is a $|\mathcal{G}^{comp}(g,t)| \times (R+1)$ matrix, and, for relevance to hold, we need that $\textrm{Rank}\big(\mathbf{\Gamma}(g,t)\big) = R+1$.  In the next proposition, we show that relevance holds under the combination of \Cref{ass:rank-delta-F,ass:rank-Lambda}.

\begin{proposition} \label{prop:ranks} 
    Under \Cref{ass:staggered,ass:no-anticipation,ass:sampling,ass:ife,ass:sel}, 
    \begin{itemize}
    \item[(1)] If both \Cref{ass:rank-delta-F,ass:rank-Lambda} hold, then 
    \begin{align*}
        \textrm{Rank}\big(\mathbf{\Gamma}(g,t)\big) = R+1 %
    \end{align*}
    \item[(2)] If either Assumption \ref{ass:rank-delta-F} or \ref{ass:rank-Lambda} is violated (in the sense that either $\textrm{Rank}\big(\bmd \mathbf{F}^{pre(g)}\big) < R$ or $\textrm{Rank}\big(\mathbf{\Lambda}^{comp}(g,t)\big) < R+1$), then
    \begin{align*}
        \textrm{Rank}\big(\mathbf{\Gamma}(g,t)\big) < R+1
    \end{align*}
    \end{itemize}
\end{proposition}

Next, let $\mathbf{W}(g,t)$ denote a $|\mathcal{G}^{comp}(g,t)| \times |\mathcal{G}^{comp}(g,t)|$ positive definite weighting matrix.  The following result shows that $ATT(g,t)$ is indeed identified under the conditions considered in this section.
\begin{theorem}  \label{thm:attgt-identification} For some group $g \in \mathcal{G}^\dagger$, and for some time period $t \in \{g, \ldots, t^{max}(g)\}$ where $t^{max}(g)$ is the largest value of $t$ such that $|\mathcal{G}^{comp}(g,t)| \geq R+1$ and under \Cref{ass:staggered,ass:no-anticipation,ass:sampling,ass:ife,ass:sel,ass:rank-delta-F,ass:rank-Lambda}, 
\begin{align}
    \begin{pmatrix}
        \theta^*(g,t) \\ F^*(g,t) 
    \end{pmatrix} = \left( \mathbf{\Gamma}(g,t)' \mathbf{W}(g,t) \mathbf{\Gamma}(g,t) \right)^{-1} \mathbf{\Gamma}(g,t)' \mathbf{W}(g,t) \E[\ell^{comp}(g,t)(Y_{t} - Y_{g-1})] \label{eqn:param-gt-identified}
\end{align}
In addition, $ATT(g,t)$ is identified%
, and it is given by
\begin{align*}
    ATT(g,t) = \E[Y_t(g) - Y_{g-1} | G=g] - \Big( \theta^*(g,t) + F^*(g,t)'\E[\Delta Y^{pre(g)} | G=g] \Big)
\end{align*}
\end{theorem}

\Cref{thm:attgt-identification} is our main identification result in the paper.  It formalizes the conditions under which group-time average treatment effects can be identified in settings with staggered treatment adoption when untreated potential outcomes are generated by an interactive fixed effects model.

To conclude this section, we discuss what sort of aggregated parameters can be recovered in our setting.  To start with, we consider a version of an event study parameter.  As earlier, we denote event time by $e=t-g$.  Then, define $\mathcal{G}^\dagger_e := \{ g \in \mathcal{G}^\dagger : g+e \leq \mathcal{T} \}$.  Then, consider the parameter
\begin{align*}
    ATT^{ES^{\dagger}}(e) := \sum_{g \in \mathcal{G}^\dagger_e} ATT(g,g+e) \P(G=g|G \in \mathcal{G}^\dagger_e)
\end{align*}
This is an event study parameter that measures the average effect of the treatment at different lengths of exposure to the treatment.  However, because $\mathcal{G}^\dagger_e \subseteq \mathcal{G}_e$ (i.e., the set of available groups for which our approach can recover group-time average treatment effects may be smaller than the full set of groups), $ATT^{ES^{\dagger}}(e)$ is, in general, not the same as $ATT^{ES}(e)$ defined earlier; nor is it possible to recover $ATT^{ES}(e)$ using our approach (except in the special case where $R=0$ --- though, as noted earlier, this case essentially reduces to a version of difference-in-differences).  That being said, if one were using alternative identification/estimation strategies where $ATT(g,t)$ for a larger set of groups and time periods and wished to compare the resulting event studies, it would be possible to apply the weights \textit{here} to the full set of $ATT(g,g+e)$ (where $g \in \mathcal{G}_e$ rather than $\mathcal{G}^\dagger_e$).  This would essentially zero-out the contribution of $ATT(g,g+e)$ for groups that are not in $\mathcal{G}^\dagger_e$.  

Next, consider an aggregation into a single overall average treatment parameter.  Towards this end, for groups $g \in \mathcal{G}^\dagger$, define
\begin{align*}
    ATT^{G^{\dagger}}(g) := \frac{1}{t^{max}(g)-g+1} \sum_{t=g}^{t^{max}(g)} ATT(g,t)
\end{align*}
where $t^{max}(g)$ is the last time period such that $|\mathcal{G}^{comp}(g,t)| \geq R+1$ (i.e., that there are enough groups that are still untreated that the identification strategy can be implemented).  This is the average treatment effect for group $g$ in all its post-treatment periods for which we are able to recover $ATT(g,t)$.  We can aggregate this into an overall treatment effect parameter by
\begin{align*}
    ATT^{O^{\dagger}} := \sum_{g \in \mathcal{G}^\dagger} ATT^{G^{\dagger}}(g) \P(G=g | G \in \mathcal{G}^\dagger)
\end{align*}
$ATT^{O^{\dagger}}$ is the average treatment effects across groups among all groups for which our approach is able to recover any group-time average treatment effects.  Like $ATT^{ES^{\dagger}}(e)$, $ATT^{O^{\dagger}}$ is not directly comparable to $ATT^O$ defined earlier when there is treatment effect heterogeneity because it does not include as many time periods or groups.  However, it does summarize the average effect of the treatment across the periods and groups for which we are able to learn about the average effect of the treatment.  Moreover, it can be compared to overall average treatment effects using alternative identification strategies by applying these weights to those group-time average treatment effects.

\subsection{Additional Discussion}

We conclude this section with several additional remarks.

\begin{remark}[Testable Implications] 

As discussed above, \Cref{ass:rank-delta-F} says that there is enough variation in the factors in pre-treatment factors to recover $R$, the true number of effective factors.  In some cases, it may be possible to detect that \Cref{ass:rank-delta-F} is violated (i.e., the number of detectable factors in pre-treatment periods for group $g$ is less than the number of effective factors $R$).  To see this, consider the example discussed in \Cref{SubSect:Baseline} (where there are four time periods,  $R=1$, $\mathcal{G}=\{3,4,\infty\}$, and the target is $ATT(3,3)$). \Cref{ass:rank-delta-F} rules out the case that $F_1 = F_2$.  However, consider an example that violates \Cref{ass:rank-delta-F} where $F_1=F_2\neq F_3$.  In this case, there is 1 factor, but, between periods 1 and 2,  groups 3, 4, and $\infty$, all have the same trends before period 3 (thus, between periods 1 and 2, it would look like there are 0 factors).  However, in this case, from period 2 to 3, if we saw groups 4 and $\infty$ trend differently, that suggests that there really is 1 factor.  We emphasize that, in this case, it is still not possible to recover $ATT(3,3)$ --- the denominator of \Cref{eqn:F3-star-baseline} is equal to 0, and, hence, $F_3^*$ is not identified.  Despite this, notice that, from \eqref{eqn:F3-explanation}, we have that 
\begin{align*}
    \Delta F_3 = \frac{\E[\Delta Y_3 | G=\infty] - \E[\Delta Y_3 | G=4]}{\E[\lambda|G=\infty] - \E[\lambda|G=4]}. 
\end{align*}
Here, although the denominator of the equation above cannot be identified from the sampling process, it is time-invariant. Thus, the temporal variation in the numerator $\E[\Delta Y_t | G=\infty] - \E[\Delta Y_t | G=4]$ is informative of the variation in $\Delta F_t$ for $t=2,3$.  And, if $\E[\Delta Y_3 | G=\infty] - \E[\Delta Y_3 | G=4] \neq 0$, it suggests that $R=1$ rather than $R=0$ which provides evidence against interpreting an estimand based on $R=0$ as $ATT(3,3)$.   It is also worth mentioning that the approach outlined here is only a partial test.  Continuing with the same example, suppose that $\E[\lambda | G=3] \neq \E[\lambda | G=4] = \E[\lambda | G=\infty]$.  In this case, groups 4 and $\infty$ would have the same trend in outcomes between periods 2 and 3 (hence, we would not find any evidence against $R=0$ from those two groups); but the trend in outcomes over time for group 3 would be different from the other groups even in the absence of the treatment (and this trend difference would wrongly be attributed to the treatment).  This issue is a challenging one to deal with (for any panel data-based approach to causal inference).  It is akin to a difference-in-differences setting where parallel trends holds in pre-treatment periods, but it is violated in post-treatment periods.%

\end{remark}

\begin{remark}[Rotation Problem]
    \label{rem:rotation-problem} Unlike much work in the literature on factor models, notice that our approach did not require any normalizations in order to deal with the so-called rotation problem for factor models.  This problem stems from any rotation of the factors leading to an observationally equivalent data generating process  (see \citet{ahn-lee-schmidt-2013} for additional discussion along these lines).  For our approach, we do not need to recover the factors, $F_t$, themselves, but rather we only need to recover $F^*(g,t)$.  Thus, we are able to side-step the rotation problem.
\end{remark}

\begin{remark} [Dealing with Observed Covariates] \label{rem:covariates}
    Observed covariates are easily incorporated in the above set-up.  \citet{callaway-karami-2023} focus on a setting with time-invariant covariates whose effect on untreated potential outcomes can vary over time; this leads to the model for untreated potential outcomes 
    \begin{align*}
        Y_{it}(0) = \theta_t + \eta_i + \lambda_i'F_t + X_i'\beta_t + e_{it}
    \end{align*}
    with $\E[e_t | \eta, \lambda, X, G] = 0$.  If the dimension of $X_i$ is $k$, then after following the same types of differencing arguments as above, there are $(R+1) + k$ parameters to identify (the same as before plus $k$ additional parameters coming from $X_i$).  One can get extra moments from adding covariates.  It immediately follows to add the following $k$ additional moment conditions
    \begin{align*}
        \mathbf{0}_k = \E[X v_i(g,t) | G \in \mathcal{G}^{comp}(g,t)]
    \end{align*}
    This effectively extends our approach to including covariates in the model.  Perhaps more interestingly (though we do not pursue it here) is that the covariates can potentially be an additional source of over-identification restrictions (which, in turn, could possibly be used to reduce the number of distinct comparison groups for a particular group and time period).  In particular, notice that the covariates actually lead to $k \times |\mathcal{G}^{comp}(g,t)|$ moment conditions
    \begin{align*}
        \mathbf{0}_{k \times |\mathcal{G}^{comp}(g,t)|} = \E[ X v_i(g,t) | G=g']_{g' \in \mathcal{G}^{comp}(g,t)}
    \end{align*}
    For example, returning to the setting considered in \Cref{SubSect:Baseline}, in period 3 there were two available comparison groups.  Introducing a single covariate gives one additional parameter to estimate, but it gives two additional moment conditions of the form above.  We leave fully exploiting these sorts of extra moment conditions from the covariates for future work.  Another popular model that includes covariates is
    \begin{align*}
        Y_{it}(0) = \theta_t + \eta_i + \lambda_i'F_t + X_{it}'\beta + e_{it}
    \end{align*}
    It is possible to deal with these time-varying covariates in much the same way as before.  After following the same differencing strategy as before, there are $k$ additional parameters to estimate, but there are also $k$ additional moment conditions that can be used to recover these parameters.
\end{remark}

\subsection{Estimation}

It is straightforward to estimate $\big(\theta^*(g,t),F^*(g,t)'\big)'$ using the sample analogue of \Cref{eqn:param-gt-identified}.  To conserve on notation in this section, define
\begin{align*}
    A_i(g) &:= \indicator{G_i=g}\begin{pmatrix}1 \\ \widetilde{\Delta Y}_i^{pre(g)} \end{pmatrix} \\
    \delta^*(g,t) &:= \Big(\theta^*(g,t),F^*(g,t)'\Big)'
\end{align*}
Applying the analogue principle to \eqref{eqn:param-gt-identified} given a positive definite matrix $\widehat{\mathbf{W}}$, the estimator of $\delta^*(g,t)$ is
\begin{align}\label{eqn:Estimator_delta}
    \widehat{\delta}^*(g,t) &= \left( \widehat{\mathbf{\Gamma}}(g,t)' \widehat{\mathbf{W}}(g,t)\widehat{\mathbf{\Gamma}}(g,t) \right)^{-1} \widehat{\mathbf{\Gamma}}(g,t)' \widehat{\mathbf{W}}(g,t) \E_n[\ell^{comp}_i(g,t)(Y_{it} - Y_{ig-1})]
\end{align} where $\E_n[\cdot]$ denotes the sample average operator across $i=1,\dots,n$ and \begin{align*}
    \widehat{\mathbf{\Gamma}}(g,t) := \E_n\left[ \ell_i^{comp}(g,t) \begin{pmatrix} 1 \\ \widetilde{\Delta Y}_i^{pre(g)} \end{pmatrix}' \right].
\end{align*} 
Let $\hat{p}_g:=\E_n[\indicator{G_i=g}]$.  By the definition of the set of all groups $\mathcal{G}$, we have that $p_g > 0$ for all $g \in \mathcal{G}$. From \Cref{thm:attgt-identification}, $ATT(g,t)$ can be re-written as
\begin{equation}
    ATT(g,t) = p_g^{-1} \left\{\E\Big[\indicator{G=g}(Y_t - Y_{g-1})\Big] - \E\Big[A(g)\Big]^\prime \delta^*(g,t) \right\} \label{eqn:attgt-est}
\end{equation} which, by the analogue and plugin principles, suggests the following estimator:
\begin{equation}
    \widehat{ATT}(g,t) = \hat{p}_g^{-1} \left\{ \E_n\Big[\indicator{G_i=g}(Y_{it} - Y_{ig-1})\Big] - \E_n\Big[A_i(g)\Big]^\prime \widehat{\delta}^*(g,t) \right\} \label{eqn:attgt-hat}
\end{equation}

We impose the following standard assumption.
\begin{assumption}\label{ass:BndConst_W} 
 $\E[||\mathbf{Y}||^4] < \infty$ and $\widehat{\mathbf{W}}(g,t) \xrightarrow{p} \mathbf{W}(g,t)$ where $\mathbf{Y}:=(Y_1,\ldots,Y_{\mathcal{T}})'$ and $\mathbf{W}(g,t)$ is positive definite.
\end{assumption}
\noindent As each element in $\ell^{comp}$ is bounded, a fourth-moment bound on $\ell^{comp}$ is thus satisfied automatically, i.e, $||\ell_i^{comp}||^4 < \infty$ by construction.  In \Cref{Lemma_deltaInfluence} in \Cref{app:proofs}, we show that 
\begin{align*}
    \sqrt{n}(\widehat{\delta}^*(g,t) - \delta^*(g,t)) = \mathbf{B}(g,t)\frac{1}{\sqrt{n}}\sum_{i=1}^{n}\ell_i^{comp}(g,t)'v_i(g,t) + o_p(1)
\end{align*}
where
\begin{align*}
    \mathbf{B}(g,t):= \left( \mathbf{\Gamma}(g,t)' \mathbf{W}(g,t)\mathbf{\Gamma}(g,t) \right)^{-1} \mathbf{\Gamma}(g,t)' \mathbf{W}(g,t)
\end{align*}
To proceed with the asymptotic distribution of $\widehat{ATT}(g,t)$, we introduce more notation. Define 
\begin{align*}
    \psi_{igt}^{(1)}&:= \frac{\indicator{G_i=g}(Y_{it} - Y_{ig-1}) - \E[\indicator{G=g}(Y_{t} - Y_{g-1})]}{p_g};\\
    \psi_{igt}^{(2)}&:= - \frac{\E[A_i(g)]}{p_g}' \mathbf{B}(g,t) \ell^{comp}_i(g,t)'v_i(g,t); \\
    \psi_{igt}^{(3)}&:= -\frac{\delta^*(g,t)'}{p_g}(A_i(g) -\E[A_i(g)]);\\
    \psi_{igt}^{(4)}&:= -\frac{ATT(g,t)}{p_g}(\indicator{G_i=g} - p_g);\quad \text{and}\\
    \psi_{igt}&:= \psi_{igt}^{(1)} + \psi_{igt}^{(2)} + \psi_{igt}^{(3)} + \psi_{igt}^{(4)}.
\end{align*} 

\begin{theorem}[Consistency and Asymptotic Normality]\label{Theorem:Asymp_Normal} Suppose \Cref{ass:staggered,ass:no-anticipation,ass:sampling,ass:ife,ass:sel,ass:rank-delta-F,ass:rank-Lambda,ass:BndConst_W} hold, then for some group $g \in \mathcal{G}^\dagger$, and for some time period $t \in \{g, \ldots, t^{max}(g)\}$ where $t^{max}(g)$ is the largest value of $t$ such that $|\mathcal{G}^{comp}(g,t)| \geq R+1$,
\begin{enumerate}[(i)]
    \item $\widehat{ATT}(g,t)$ is asymptotically linear, and it satisfies the relation
 \begin{align*}
     \sqrt{n}(\widehat{ATT}(g,t) - ATT(g,t)) = \frac{1}{\sqrt{n}}\sum_{i=1}^n \psi_{igt} + o_p(1);
\end{align*}
    \item $\widehat{ATT}(g,t) \xrightarrow{p} ATT(g,t) $ as $n \rightarrow \infty$ for each pair $(g,t)$ 
    \item In addition, \[\sqrt{n}(\widehat{ATT}(g,t) - ATT(g,t)) \xrightarrow{d} \mathcal{N}(0,\sigma_{gt}^2) \] where $\sigma_{gt}^2 = \E[\psi_{igt}^2] $.
\end{enumerate}
\end{theorem}

Next, we move to establishing the joint asymptotic normality of $ATT(g,t)$ across all the periods and groups for which it is identified and to provide a uniform inference procedure as well as the limiting distributions of aggregated treatment effect parameters such as the event study and overall average treatment effect discussed above.  Recall that all of the aggregated treatment effect parameters discussed above can be written as and estimated by
\begin{align*}
    \theta = w'ATT = \sum_{g \in \mathcal{G}^{\dagger}} \sum_{t=g}^{t^{max}(g)} w(g,t) ATT(g,t) \quad \textrm{and} \quad \hat{\theta} = \hat{w}'\widehat{ATT} = \sum_{g \in \mathcal{G}^{\dagger}} \sum_{t=g}^{t^{max}(g)} \hat{w}(g,t) \widehat{ATT}(g,t),
\end{align*}
respectively, where $\theta$ is generic notation for an aggregated treatment effect parameter, $ATT$ is a vector that stacks $ATT(g,t)$ for $(g,t) \in \mathcal{G}^\dagger \times \{g, \ldots, t^{max}(g)\}$, $w$ is a vector that stacks weights on each $ATT(g,t)$ to aggregate them into $\theta$; and $\hat{\theta}$, $\hat{w}$, and $\widehat{ATT}$ are the corresponding estimators.  Also, notice that, in practice, the $w(g,t)$'s (the weights on particular group-time average treatment effects) need to be estimated, and, therefore, we should take into account their estimation effect.  Given the conditions discussed above, for all of the aggregated parameters that we consider, the weights are asymptotically linear, which we write generically as 
\begin{align*}
    \sqrt{n}(\hat{w} - w) = \frac{1}{\sqrt{n}}\sum_{i=1}^n \mathcal{W}_i + o_p(1)  
\end{align*}
with $\E[\mathcal{W}_i]=0$.  As a last piece of additional notation, define 
$\Psi_i$ as the vector comprising the set $\{\psi_{igt}\}_{(g,t)\in \mathcal{G}^\dagger \times \{g, \ldots, t^{max}(g)\}}$. %
The following result provides an asymptotic linear representation and an asymptotic normality result.
\begin{proposition}\label{Corollary:ATT}
    Under \Cref{ass:staggered,ass:no-anticipation,ass:sampling,ass:ife,ass:sel,ass:rank-delta-F,ass:rank-Lambda,ass:BndConst_W}, the estimator $\widehat{ATT}$ has the following asymptotic linear representation:
    \[\sqrt{n}(\widehat{ATT} - ATT) = \frac{1}{\sqrt{n}}\sum_{i=1}^n \Psi_i + o_p(1).\] Further,
    \begin{align*}
        &\sqrt{n}(\widehat{ATT} - ATT) \xrightarrow{d} \mathcal{N}(0,\E[\Psi\Psi^\prime]) \text{; and} \\
        &\sqrt{n}(\hat{w}'\widehat{ATT} - w'ATT) = \frac{1}{\sqrt{n}} \sum_{i=1}^n (w'\Psi_i + ATT'\mathcal{W}_i) + o_p(1) \xrightarrow{d} \mathcal{N}(0,\sigma_w^2)
    \end{align*}   where $\sigma_w^2:=\E[(w'\Psi_i + ATT'\mathcal{W}_i)^2]$.
\end{proposition}

\Cref{Theorem:Asymp_Normal} and \Cref{Corollary:ATT} establish the joint limiting distribution of all of the identified group-time average treatment effects as well as the aggregated treatment effect parameters.  These results can be used as the basis for conducting inference by estimating $\E[\Psi \Psi']$ and/or $\sigma^2_w$.  Instead of taking this approach, we opt for conducting inference using a multiplier bootstrap procedure that builds on the previous result. For brevity, we focus on inference for the group-time average treatment effects, but analogous results hold for the aggregated treatment effect parameters given the results above.  The multiplier bootstrap involves perturbing the influence function, and offers a number of advantages relative to the nonparametric bootstrap in terms of computational speed.  To fix ideas, consider some $iid$ mean-zero $\zeta_i$ with unit variance and a finite third moment that is drawn independently of the data, e.g., the standard normal or a binary $(-1,1)$ each with probability 1/2. The bootstrap estimate is obtained using
\begin{align*}
\widehat{ATT}^\star = \widehat{ATT} + \frac{1}{n}\sum_{i=1}^n \zeta_i \widehat{\Psi}_i
\end{align*}
where $\widehat{\Psi}_i$ is an estimate of $\Psi_i$ that replaces population parameters with estimates.  See \citet{kline-santos-2012} for more discussion of the multiplier bootstrap.

The next corollary establishes the (asymptotic) validity of the multiplier bootstrap procedure discussed above.
\begin{corollary} \label{prop:bootstrap-validity} Under \Cref{ass:staggered,ass:no-anticipation,ass:sampling,ass:ife,ass:sel,ass:rank-delta-F,ass:rank-Lambda,ass:BndConst_W},
    \begin{align*}
        \sqrt{n}\Big(\widehat{ATT}^* - \widehat{ATT}\Big) \xrightarrow{d^*} \N(0, \E[\Psi \Psi'])
    \end{align*}
    where $\xrightarrow{d^*}$ denotes convergence in bootstrap distribution (conditional on the original data).
\end{corollary}

\Cref{prop:bootstrap-validity} says that, under standard conditions and conditional on the observed data, $\sqrt{n}(\widehat{ATT}^\star - \widehat{ATT})$ has the same asymptotic distribution as $\sqrt{n}(\widehat{ATT} - ATT)$ in \Cref{Corollary:ATT}. For each $(g,t)\in (g,t) \in \mathcal{G}^\dagger \times \{g, \ldots, t^{max}(g)\}$, one can construct standard errors for $\widehat{ATT}(g,t)$ using $\sigma_{gt}^\star = (q_{0.75}(g,t) - q_{0.25}(g,t))/(z_{0.75} - z_{0.25}) $ where $q_\tau(g,t)$ denotes the $\tau$'th quantile from the empirical distribution of $\sqrt{n}(\widehat{ATT}^\star(g,t) - \widehat{ATT}(g,t))$ and $z_\tau$ denotes the $\tau$'th quantile of the standard normal distribution.

\section{Simulations}

In this section, we provide Monte Carlo simulations to illustrate the finite sample properties of our proposed estimation strategy.  And, in particular, we compare our approach to the one in \citet{callaway-karami-2023} and to difference-in-differences and unit specific linear-trends approaches.  We generate untreated potential outcomes by
\begin{align*}
    Y_{it}(0) = \theta_t + \eta_i + \lambda_i'F_t + Z_i'\beta + u_{it}
\end{align*}
This corresponds to the model in \Cref{eqn:ife-reg} except for the term $Z_i'\beta$.  In our case, because this term is time-invariant, it is absorbed into the unit-fixed effect $\eta_i$, but this is the key term in the estimation strategy proposed in \citet{callaway-karami-2023}.  Following the simulations in \citet{callaway-karami-2023}, we consider the case where $\lambda_i$ and $Z_i$ are both vectors containing three elements.  We also consider the case where $\mathcal{T}=8$ and $\mathcal{G} = \{5,6,7,8,\infty\}$ (so that we have units that become treated in periods 5, 6, 7, 8, and some units that remain untreated in all periods).  We assign units to each group with equal probability.  For $j=1,2,3$, we take $Z_{ij} \sim_{iid} \N(0,1)$.  We also set $\beta_j = 0$ (which implies that it does not actually directly affect untreated potential outcomes though it can still be an additional source of moment conditions for the approach in \citet{callaway-karami-2023}).  Next, for $h \in \{0,1\}$, we set $\eta_i \sim h G_i + \varepsilon_{\eta,i}$ where $\varepsilon_{\eta,i} \sim \N(0,0.1)$.  We vary $h$ across simulations --- when $h=0$, the distribution of $\eta_i$ is the same across groups, but otherwise it is different. We set $\lambda_{i1} = 1 + 2G_i + \rho_1 Z_{i1} + \varepsilon_{i1}$, $\lambda_{i2} = 1 - 5G_i + \rho_2 Z_{i2} + \varepsilon_{i2}$, and $\lambda_{i3} = 5 - 10G_i + \rho_3Z_{i3} + \varepsilon_{i3}$ where $\varepsilon_j \sim_{iid} \N(0,1)$ for $j \in \{1,2,3\}$.  For all the simulations reported below, we set $\rho_j = 0.2$ for $j=1,2,3,$ which corresponds to the ``medium strength instrument'' case considered in \citet{callaway-karami-2023}.  Next, we set $\tilde{F}_{1t} = t$, $\tilde{F}_{2t} = (-1)^t t \log(t)$ and $\tilde{F}_{3t} = (-1)^{\indicator{t > 5}} (5-|5-t|)^2$.  We consider four types of designs below.  In designs labeled ``no unobs.\,het.'', we set $F_t=(0,0,0)'$ and $h=0$; in designs that labeled ``0 IFE'', we set $F_t=(0,0,0)'$ and $h=1$; for ``1 IFE'', we set $F_t=(\tilde{F}_{1t}, 0, 0)'$ and $h=1$; for ``2 IFE'', we set $F_t=(\tilde{F}_{1t}, \tilde{F}_{2t}, 0)'$ and $h=1$.  We (with a few exceptions) provide results for the overall average treatment effect.  For all the simulations below, $n=1000$, the results are based on 1000 Monte Carlo simulations.  Finally, in all simulations, we set $Y_{it}(1) = Y_{it}(0)$, so the true value of all treatment effect parameters is equal to 0.

\begin{table}
\caption{Monte Carlo Simulations Comparing Estimators}
\label{tab:mc-sims1}
\centering
\resizebox{\linewidth}{!}{
\begin{tabular}[t]{rrrrrrrrrrrrrrrrr}
\toprule
\multicolumn{17}{c}{Truth} \\
\cmidrule(l{3pt}r{3pt}){1-17}
\multicolumn{1}{c}{ } & \multicolumn{4}{c}{-1 IFE} & \multicolumn{4}{c}{0 IFE} & \multicolumn{4}{c}{1 IFE} & \multicolumn{4}{c}{2 IFE} \\
\cmidrule(l{3pt}r{3pt}){2-5} \cmidrule(l{3pt}r{3pt}){6-9} \cmidrule(l{3pt}r{3pt}){10-13} \cmidrule(l{3pt}r{3pt}){14-17}
 & Bias & RMSE & MAD & Rej. & Bias & RMSE & MAD & Rej. & Bias & RMSE & MAD & Rej. & Bias & RMSE & MAD & Rej.\\
\midrule
\addlinespace[0.3em]
\multicolumn{17}{l}{\textbf{Staggered IFE}}\\
\hspace{1em}IFE=-1 & 0.0017 & 0.0506 & 0.0367 & 0.064 & 4.8491 & 4.8537 & 4.8483 & 1.000 & 79.2864 & 79.3459 & 79.3705 & 1.000 & 411.6776 & 411.8159 & 412.0379 & 1.000\\
\hspace{1em}IFE=0 & 0.0043 & 0.0657 & 0.0448 & 0.059 & 0.0017 & 0.0637 & 0.0428 & 0.051 & 16.2459 & 16.2648 & 16.2357 & 1.000 & 261.5271 & 262.0339 & 261.1525 & 1.000\\
\hspace{1em}IFE=1 & 0.1496 & 2.9402 & 0.0737 & 0.030 & 0.0003 & 1.4413 & 0.0783 & 0.030 & 0.0012 & 0.0804 & 0.0517 & 0.141 & 0.0067 & 0.1512 & 0.1028 & 0.154\\
\hspace{1em}IFE=2 & -0.0477 & 2.8651 & 0.1050 & 0.032 & 23.8781 & 756.1583 & 0.1091 & 0.022 & 0.7073 & 21.0779 & 0.0922 & 0.037 & 0.2740 & 6.7571 & 0.1424 & 0.051\\
\addlinespace[0.3em]
\multicolumn{17}{l}{\textbf{CK-2023}}\\
\hspace{1em}IFE=0 & 0.0042 & 0.0657 & 0.0451 & 0.054 & 0.0016 & 0.0637 & 0.0432 & 0.053 & 14.9408 & 14.9435 & 14.9368 & 1.000 & 269.2484 & 269.4390 & 269.3669 & 1.000\\
\hspace{1em}IFE=1 & 0.0017 & 0.0726 & 0.0482 & 0.022 & -0.0015 & 0.0783 & 0.0518 & 0.030 & 5.1226 & 7.2105 & 4.9250 & 0.325 & 17.9017 & 28.4761 & 19.3960 & 0.235\\
\hspace{1em}IFE=2 & -0.0004 & 0.1331 & 0.0627 & 0.003 & -0.0021 & 0.1338 & 0.0610 & 0.004 & 1.4838 & 9.5252 & 3.3613 & 0.020 & 1.2997 & 18.3162 & 7.8324 & 0.010\\
\addlinespace[0.3em]
\multicolumn{17}{l}{\textbf{Difference-in-differences}}\\
\hspace{1em}& 0.0014 & 0.0667 & 0.0465 & 0.060 & -0.0009 & 0.0678 & 0.0457 & 0.071 & 14.9470 & 14.9499 & 14.9547 & 1.000 & 268.8663 & 269.0874 & 268.8677 & 1.000\\
\addlinespace[0.3em]
\multicolumn{17}{l}{\textbf{Linear Trends}}\\
\hspace{1em} & 0.0031 & 0.1442 & 0.0969 & 0.062 & 0.0002 & 0.1437 & 0.0967 & 0.064 & -0.0013 & 0.1467 & 0.0990 & 0.069 & 142.7103 & 147.1414 & 142.8595 & 1.000\\
\bottomrule
\end{tabular}}
\begin{justify}
    \small \textit{Notes: } The table provides simulation results that come from using the Staggered IFE approach proposed in the current paper, from \citet{callaway-karami-2023} (labeled CK-2023), from difference-in-differences, and from unit-specific linear trends estimation strategies.  All results in the table are provided for $ATT^O$.  The data generating process is described in the main text.  Columns labeled ``Bias'', ``RMSE'', ``MAD'', and ``Rej.'' report estimated bias, root mean squared error, median absolute deviation, and rejection rate for a test that $ATT(5,5)=0$ (which is true) at the 5\% significance level.  There are four sets of columns that differ based on the number of interactive fixed effects in the data generating process; the set of columns labeled ``no unobs.\,het'' have the distribution of $\eta_i$, which is the unit fixed effect, also to be the same across groups.  The row labeled ``no unobs.\,het'' uses the staggered IFE approach based on comparing levels of outcomes across groups.  The rows labeled ``IFE=j'' for $j \in \{0,1,2\}$ differ based on the number of interactive fixed effects included in the estimation strategy in that row.
\end{justify}
\end{table}

The first set of results is provided in \Cref{tab:mc-sims1}.  To start with, notice that difference-in-differences and linear trends approaches seem to perform well in the cases where they are expected to perform well and perform poorly in expected cases --- for difference-in-differences these cases are with no unobserved heterogeneity and when there are 0 interactive fixed effects; for linear trends, these cases are when there is no unobserved heterogeneity, when there are 0 interactive fixed effects, or with 1 interactive fixed effect (this case holds because the first factor is linear in the DGP that we consider here).  Similarly, \citet{callaway-karami-2023} performs roughly as expected.  The estimation strategy generally performs well when the correct number of interactive fixed effects are included in the model.  Including too few interactive fixed effects results in very poor performance of the estimation strategy.  On the other hand, including too many interactive fixed effects tends to result in somewhat higher RMSE and MAD (as well as conservative inference).  

As for our Staggered IFE approach, it performs somewhat better than  \citet{callaway-karami-2023} --- when the correct number of interactive fixed effects are included, our approach performs well.  It has lower mean squared error than CK-2023 in all specifications.  In the case with 2 interactive fixed effects, the root mean squared error is notably large though this seems to be driven by a small number of simulations where the estimator performs poorly as the bias is low, the median absolute deviation is small, and the estimator appears to control size well. In the case with one interactive fixed effect, our estimator over-rejects, but in the other cases, when the number of interactive fixed effects is correct, our inference procedure appears to work well.

It is also interesting to consider settings where the researcher includes the wrong number of interactive fixed effects.  First, including too few interactive fixed effects can result in very poor performance (one notable exception to this in the results in the table occurs in the case where we included only 1 interactive fixed effect but the true number of interactive fixed effects was 2).  On the other hand, including too many tended to mainly result in conservative inference rather than biased estimates (in some cases the bias and RMSE are large, but the estimator typically performs well in terms of median absolute deviation).  We leave a formal analysis of the implications of including too many interactive fixed effects to future work, but these simulations suggest an asymmetry between including too many interactive fixed effects relative to too few interactive fixed effects similar to what is noted in \citet{moon-weidner-2015} in a somewhat different context.

Next, we provide additional simulation results that vary the differences between groups (in terms of the mean of their interactive fixed effects).  We use largely the same data-generating process as above except that we only consider the case where the number of interactive fixed effects is exactly equal to 1.  We also set $\lambda_{i1} = 1 + G_i l + \epsilon_{i1}$.  We vary $l$ among $\{0.5, 0.1, 0.01, 0.001\}$.  Smaller values of $l$ indicate that the groups are more similar to each other in terms of their interactive fixed effects.  For these simulations, we only provide results using the Staggered IFE approach considered in the current paper.  We vary the number of interactive fixed effects included in the model and report results for the overall average treatment effect as well as event study parameters for $e=0,1,2$.

\begin{table}
\caption{Monte Carlo Results Varying Differences between Groups}
\label{tab:mc-sims2}
\centering
\resizebox{\linewidth}{!}{
\begin{tabular}[t]{rrrrrrrrrrrrrrrrr}
\toprule
\multicolumn{17}{c}{Truth} \\
\cmidrule(l{3pt}r{3pt}){1-17}
\multicolumn{1}{c}{ } & \multicolumn{4}{c}{l=0.5} & \multicolumn{4}{c}{l=0.1} & \multicolumn{4}{c}{l=0.01} & \multicolumn{4}{c}{l=0.001} \\
\cmidrule(l{3pt}r{3pt}){2-5} \cmidrule(l{3pt}r{3pt}){6-9} \cmidrule(l{3pt}r{3pt}){10-13} \cmidrule(l{3pt}r{3pt}){14-17}
 & Bias & RMSE & MAD & Rej. & Bias & RMSE & MAD & Rej. & Bias & RMSE & MAD & Rej. & Bias & RMSE & MAD & Rej.\\
\midrule
\addlinespace[0.3em]
\multicolumn{17}{l}{\textbf{Overall}}\\
\hspace{1em}IFE=-1 & 23.1995 & 23.2183 & 23.2519 & 1.000 & 8.1652 & 8.1823 & 8.1618 & 1.000 & 4.8728 & 4.8971 & 4.8577 & 1.000 & 4.5252 & 4.5527 & 4.5310 & 1.000\\
\hspace{1em}IFE=0 & 3.9347 & 3.9394 & 3.9299 & 1.000 & 0.7635 & 0.7753 & 0.7677 & 1.000 & 0.0811 & 0.1529 & 0.1082 & 0.131 & 0.0041 & 0.1269 & 0.0883 & 0.067\\
\hspace{1em}IFE=1 & -0.0007 & 0.0790 & 0.0519 & 0.122 & 0.0003 & 0.0795 & 0.0545 & 0.116 & 0.0918 & 3.0369 & 0.0837 & 0.050 & 0.0035 & 0.8259 & 0.0857 & 0.037\\
\hspace{1em}IFE=2 & 0.0046 & 0.9781 & 0.0984 & 0.040 & 0.2048 & 5.4894 & 0.1040 & 0.049 & 0.0887 & 1.9944 & 0.1197 & 0.027 & -0.1608 & 12.5335 & 0.1326 & 0.020\\
\addlinespace[0.3em]
\multicolumn{17}{l}{\textbf{Event Study: $\bm{e=0}$}}\\
\hspace{1em}IFE=-1 & 18.4084 & 18.4456 & 18.4850 & 1.000 & 6.1535 & 6.1751 & 6.1565 & 1.000 & 3.5824 & 3.6064 & 3.5731 & 1.000 & 3.3141 & 3.3402 & 3.3224 & 1.000\\
\hspace{1em}IFE=0 & 1.7981 & 1.8028 & 1.7950 & 1.000 & 0.3046 & 0.3149 & 0.3059 & 0.975 & 0.0322 & 0.0825 & 0.0601 & 0.101 & 0.0029 & 0.0759 & 0.0528 & 0.072\\
\hspace{1em}IFE=1 & 0.0010 & 0.0809 & 0.0563 & 0.097 & 0.0005 & 0.0812 & 0.0564 & 0.079 & 0.1006 & 3.0498 & 0.0812 & 0.044 & -0.0020 & 0.7186 & 0.0899 & 0.040\\
\hspace{1em}IFE=2 & 0.0228 & 0.9729 & 0.1057 & 0.044 & 0.1803 & 4.8735 & 0.1092 & 0.039 & 0.0856 & 1.9058 & 0.1256 & 0.020 & 0.0318 & 1.8626 & 0.1427 & 0.022\\
\addlinespace[0.3em]
\multicolumn{17}{l}{\textbf{Event Study: $\bm{e=1}$}}\\
\hspace{1em}IFE=-1 & 17.3523 & 17.4121 & 17.4418 & 1.000 & 5.9361 & 5.9696 & 5.9471 & 1.000 & 3.4808 & 3.5220 & 3.4749 & 1.000 & 3.2279 & 3.2711 & 3.2360 & 1.000\\
\hspace{1em}IFE=0 & 3.5452 & 3.5585 & 3.5587 & 1.000 & 0.6269 & 0.6453 & 0.6313 & 0.991 & 0.0670 & 0.1566 & 0.1074 & 0.127 & 0.0014 & 0.1382 & 0.0972 & 0.083\\
\hspace{1em}IFE=1 & -0.0040 & 0.0940 & 0.0666 & 0.090 & -0.0010 & 0.0970 & 0.0637 & 0.120 & -0.0361 & 0.8387 & 0.0927 & 0.041 & 0.0133 & 0.8836 & 0.0953 & 0.028\\
\hspace{1em}IFE=2 & -0.0417 & 1.4827 & 0.1336 & 0.037 & 0.1024 & 10.8848 & 0.1271 & 0.033 & 0.0159 & 3.8906 & 0.1553 & 0.028 & -0.7863 & 50.8690 & 0.1717 & 0.024\\
\addlinespace[0.3em]
\multicolumn{17}{l}{\textbf{Event Study: $\bm{e=2}$}}\\
\hspace{1em}IFE=-1 & 17.2294 & 17.2918 & 17.2922 & 1.000 & 6.2417 & 6.2803 & 6.2506 & 1.000 & 3.7858 & 3.8379 & 3.8000 & 1.000 & 3.5398 & 3.5968 & 3.5532 & 1.000\\
\hspace{1em}IFE=0 & 5.2676 & 5.2884 & 5.2709 & 1.000 & 1.0731 & 1.1034 & 1.0886 & 0.991 & 0.1146 & 0.2699 & 0.1878 & 0.110 & 0.0000 & 0.2379 & 0.1641 & 0.076\\
\hspace{1em}IFE=1 & 0.0011 & 0.1485 & 0.0978 & 0.108 & 0.0091 & 0.1624 & 0.1083 & 0.098 & 0.0188 & 1.9509 & 0.1673 & 0.028 & -0.0120 & 1.5785 & 0.1743 & 0.025\\
\bottomrule
\end{tabular}}
\begin{justify}
    \small \textit{Notes: } The table provides simulation results that vary the difference between $\E[\lambda|G=g]$ across groups (as discussed in the text) using the Staggered IFE approach proposed in the paper.  Columns labeled ``Bias'', ``RMSE'', ``MAD'', and ``Rej.'' report estimated bias, root mean squared error, median absolute deviation, and rejection rate for a test that the $ATT$ parameters are equal to 0 (which is true) at the 5\% significance level.  There are four sets of columns that differ based on the relative differences in the mean of $\lambda_i$ across groups.  The three sets of columns differ based on the treatment effect parameter that they report; the set of columns labeled ``Overall'' provides results for the overall $ATT$, while the sets of columns labeled ``Event Study'' provide results for event study parameters when the length of exposure to the treatment is 0 (i.e., the on-impact effect of the treatment) and when the length of exposure to the treatment is 1, respectively.  The rows labeled ``IFE=j'' for $j \in \{-1,0,1,2\}$ differ based on the number of interactive fixed effects included in the estimation strategy in that row.
\end{justify}
\end{table}

These results are provided in  \Cref{tab:mc-sims2}.  The results are most interesting for $l=0.01$ and $l=0.001$.  These cases essentially correspond to very small violations of parallel trends, and this is an empirically relevant Monte Carlo simulation where a researcher might be unsure about whether or not the parallel trends assumption holds.  It is especially interesting to compare the rows labeled $IFE=0$ (this estimate is essentially a GMM version of difference-in-differences where no interactive fixed effects are included in the model) and $IFE=1$ (which is the correctly specified model).    Here, the RMSE and MAD are sometimes smaller when no interactive fixed effects are included in the model.  This is especially true for $e=0$ and somewhat less true (particularly for MAD) for $e=1$ or $e=2$.  On the other hand, including too few interactive fixed effects tends to result in over-rejection.  This suggests an interesting trade-off in terms of including more or less interactive fixed effects, and it is also interesting to note that, in cases where the violation of parallel trends is small, inappropriately imposing parallel trends performs better closer to the time period when the treatment is implemented relative to periods further away from the time period when the treatment is implemented.  This makes some sense too as the violations of parallel trends here accumulate over more and more periods as units move further away from the time period in which they were treated.

\section{Conclusion} 

In this paper, we have introduced a new approach to recovering causal effect parameters in a setting where untreated potential outcomes are generated by an interactive fixed effects model,  where the researcher only has access to a few periods of data, and where there is staggered treatment adoption.  In our view, interactive fixed effects models are attractive in a large number of applications in economics where it is not clear whether or not it is reasonable to think that the effect of time-invariant unobservables is constant over time (this is evidenced by the ubiquity of pre-testing in difference-in-differences applications in economics).  However, their adoption has, at least arguably, been hampered due to either requiring a large number of time periods or relatively strong auxiliary assumptions in order to estimate.  The approach that we proposed in the current paper works without requiring a large number of periods or extra assumptions.  For the vast majority of empirical applications in economics, the only additional requirement for using our approach is that there is variation in treatment timing across different units.  This is common in many applications, and, in fact, has often been framed as a ``complication'' in many applications; instead, we argue that variation in treatment timing may present a useful opportunity to explore more complicated models in treatment effect applications with panel data.

\pagebreak
\printbibliography

\pagebreak

\appendix
\section{Theoretical Results} \label{app:proofs}

For the results below, we use the following results on matrix products below. %
For a $p\times q$ matrix $\mathbf{A}$, a $q \times r$ matrix $\mathbf{B}$, and an $r \times s$ matrix $\mathbf{C}$,  
\begin{align}
    \textrm{Rank}(\mathbf{AB}) \leq \min\{\textrm{Rank}(\mathbf{A}), \textrm{Rank}(\mathbf{B}) \} \label{eqn:matrix-rank-inequality-1}
\end{align}
and that
\begin{align} \label{eqn:matrix-rank-inequality-2}
    \textrm{Rank}(\mathbf{ABC}) \leq \min\{\textrm{Rank}(\mathbf{A}), \textrm{Rank}(\mathbf{B}), \textrm{Rank}(\mathbf{C}) \}
\end{align}
See \citet[4.15(b)]{abadir-magnus-2005} for \Cref{eqn:matrix-rank-inequality-1}, and note that \Cref{eqn:matrix-rank-inequality-2} follows from iterating \Cref{eqn:matrix-rank-inequality-1}.  
If, in addition, $q=r$ and $\textrm{Rank}(\mathbf{B}) = q$ (so that $\mathbf{B}$ is a square $q\times q$ matrix with full rank and $\mathbf{AB}$ is a $p \times q$ matrix), then
\begin{align} \label{eqn:matrix-rank-equality-2}
    \textrm{Rank}(\mathbf{AB}) = \textrm{Rank}(\mathbf{A})
\end{align}
See \citet[4.24(b)]{abadir-magnus-2005} for a proof.

\begin{proof}[\textbf{Proof of \Cref{prop:ranks}}]

To start with, we define some additional notation.  Let
\begin{align*}
    \widetilde{\mathbf{\Gamma}}(g,t):= \E\Big[\Big({\widetilde{\Delta Y}}^{{pre(g)}^{'}} - \E\big[{\widetilde{\Delta Y}}^{{pre(g)}^{'}}|G\in \mathcal{G}^{comp}(g,t) \big]\Big) \Big|G=g' \Big]_{g' \in \mathcal{G}^{comp}(g,t)}
\end{align*}
and, analogously,
\begin{align*}
    \widetilde{\mathbf{\Lambda}}(g,t):= \E\Big[\Big(\lambda - \E\big[\lambda|G\in \mathcal{G}^{comp}(g,t) \big]\Big) \Big|G=g' \Big]_{g' \in \mathcal{G}^{comp}(g,t)}
\end{align*}
which are demeaned versions of $\mathbf{\Gamma}(g,t)$ and $\mathbf{\Lambda}(g,t)$.  Note that, since $\mathbf{\Gamma}(g,t)$ and $\mathbf{\Lambda}(g,t)$ both include an intercept, we have that $\textrm{Rank}(\mathbf{\Gamma}(g,t)) = \textrm{Rank}(\tilde{\mathbf{\Gamma}}(g,t)) + 1$, and, likewise that $\textrm{Rank}(\mathbf{\Lambda}(g,t)) = \textrm{Rank}(\tilde{\mathbf{\Lambda}}(g,t)) + 1$.  

Notice that, from the definition of $\ell^{comp}(g,t)$ and \Cref{ass:sel}, we have that
\begin{align}
    \widetilde{\mathbf{\Gamma}}(g,t) &= \E \Big[\Big(\mathbf{\Omega}'\bmd \mathbf{F}^{pre(g)}\big(\lambda- \E\big[\lambda|G\in \mathcal{G}^{comp}(g,t) \big]\big)\Big)'\Big|G=g'\Big]_{g' \in \mathcal{G}^{comp}(g,t)} \nonumber \\
    &= \tilde{\mathbf{\Lambda}}(g,t) {\bmd \mathbf{F}}^{{pre(g)}^{'}} \mathbf{\Omega} \label{eqn:rank-tilde-gamma}
\end{align}

To prove the first part of the proposition (i.e., the case where \Cref{ass:rank-delta-F,ass:rank-Lambda} both hold), notice  that $\mathrm{Rank}\big({\bmd \mathbf{F}}^{{pre(g)}^{'}} \mathbf{\Omega}\big)=\textrm{Rank}\big( \mathbf{\Omega}'\bmd \mathbf{F}^{pre(g)} \big) = \mathrm{Rank}\big({\bmd \mathbf{F}}^{pre(g)}\big)=R$ (where the last equality holds by \Cref{ass:rank-delta-F}).  Then, from \Cref{eqn:matrix-rank-equality-2}, we have that $\textrm{Rank}(\tilde{\mathbf{\Gamma}}(g,t)) = \textrm{Rank}(\tilde{\mathbf{\Lambda}}(g,t)) = R$ (where the last equality holds by \Cref{ass:rank-Lambda}).  %

For the second part of the proposition (where at least one of Assumptions \ref{ass:rank-delta-F} or \ref{ass:rank-Lambda} does not hold), 
\begin{align*}
    \textrm{Rank}(\tilde{\mathbf{\Gamma}}(g,t)) \leq \min\{ \textrm{Rank}(\tilde{\mathbf{\Lambda}}(g,t)), \textrm{Rank}(\bmd \mathbf{F}^{{pre(g)}^{'}}) \} < R
\end{align*}
where the first inequality holds from \Cref{eqn:matrix-rank-inequality-2} (and because $\mathbf{\Omega}$ has rank $R$), and the second inequality holds because at least one of Assumptions \ref{ass:rank-delta-F} or \ref{ass:rank-Lambda} does not hold in this case.

\end{proof}

\begin{proof}[\textbf{Proof of \Cref{thm:attgt-identification}}]
    Starting from \Cref{eqn:mom-conds} in the main text, which we derived under \Cref{ass:staggered,ass:no-anticipation,ass:sampling,ass:ife,ass:sel}, we have that
    \begin{align*}
         \mathbf{\Gamma}(g,t) \begin{pmatrix} \theta^*(g,t) \\ F^*(g,t) \end{pmatrix} = \E\Big[ \ell^{comp}(g,t)(Y_t - Y_{g-1}) \Big]
    \end{align*}
    which implies that
    \begin{align*}
        \mathbf{\Gamma}(g,t)'\mathbf{W}(g,t) \mathbf{\Gamma}(g,t) \begin{pmatrix} \theta^*(g,t) \\ F^*(g,t) \end{pmatrix} = \mathbf{\Gamma}(g,t)'\mathbf{W}(g,t) \E\Big[ \ell^{comp(g,)}(Y_t - Y_{g-1}) \Big]
    \end{align*}
    which implies that
    \begin{align*}
        \begin{pmatrix}
        \theta^*(g,t) \\ F^*(g,t) 
        \end{pmatrix} = \left( \mathbf{\Gamma}(g,t)' \mathbf{W}(g,t) \mathbf{\Gamma}(g,t) \right)^{-1} \mathbf{\Gamma}(g,t)' \mathbf{W}(g,t) \E[\ell^{comp}(g,t)(Y_{t} - Y_{g-1})]
    \end{align*}
    where this equation uses \Cref{prop:ranks}(thanks to \Cref{ass:rank-Lambda,ass:rank-delta-F}) and \Cref{ass:BndConst_W}.  Given this result, the second part of the theorem holds directly from \Cref{eqn:attgt-identified}.   This completes the proof.
\end{proof}

\bigskip

Next, we move to proving \Cref{Theorem:Asymp_Normal}.  A useful first step is to obtain an influence function representation of $\widehat{\delta}^*(g,t)$ which we provide in the following lemma.

\begin{lemma}\label{Lemma_deltaInfluence} Under \Cref{ass:staggered,ass:no-anticipation,ass:sampling,ass:ife,ass:sel,ass:rank-delta-F,ass:rank-Lambda,ass:BndConst_W}, $\widehat{\delta}^*(g,t)$ has the following influence function representation:
\[\sqrt{n}(\widehat{\delta}^*(g,t) - \delta^*(g,t)) = \mathbf{B}(g,t)\frac{1}{\sqrt{n}}\sum_{i=1}^{n}\ell_i^{comp}(g,t)'v_i(g,t) + o_p(1).\]
\end{lemma}
\begin{proof}
    Substituting the expression \eqref{eqn:mp-estimating-equation} into \eqref{eqn:Estimator_delta} gives
    \begin{align*}
      \widehat{\delta}^*(g,t) - \delta^*(g,t) = \widehat{\mathbf{B}}(g,t)\E_n[\ell_i^{comp}(g,t)'v_i(g,t)]
    \end{align*}where $\widehat{\mathbf{B}}(g,t):= \left( \widehat{\mathbf{\Gamma}}(g,t)'\widehat{\mathbf{W}}(g,t)\widehat{\mathbf{\Gamma}}(g,t) \right)^{-1} \widehat{\mathbf{\Gamma}}(g,t)'\widehat{\mathbf{W}}(g,t)$. Under \Cref{ass:rank-delta-F,ass:rank-Lambda}, $\mathbf{\Gamma}(g,t)$ has full column rank (see \Cref{prop:ranks}), thus in addition to \Cref{ass:BndConst_W}, $\left( \mathbf{\Gamma}(g,t)' \mathbf{W}(g,t) \mathbf{\Gamma}(g,t) \right)$ is non-singular. Under \Cref{ass:sampling,ass:BndConst_W,ass:sel,ass:ife,ass:rank-delta-F,ass:rank-Lambda}, $\widehat{\mathbf{B}}(g,t) = \mathbf{B}(g,t) + o_p(1) $ by the continuity of the inverse at a non-singular matrix, the continuity of the product of matrices, and the continuous mapping theorem. The stochastic boundedness of $||\ell^{comp}(g,t)||^4$ holds by construction as $\ell^{comp}(g,t) \in \mathbb{R}^{|\mathcal{G}^{comp}(g,t)|} $ comprises variables that are bounded. $\E_n[\ell_i^{comp}(g,t)'v_i(g,t)] = O_p(n^{-1/2})$ under \Cref{ass:sampling,ass:BndConst_W,ass:sel,ass:ife} and Markov's inequality. This proves the assertion as claimed.
\end{proof}

\begin{proof}[\textbf{Proof of \Cref{Theorem:Asymp_Normal}}]
    Consider the following decomposition:
    \begin{align*}
    \sqrt{n}\Big(\widehat{ATT}(g,t) - ATT(g,t)\Big) 
    &= \sqrt{n}\left\{\hat{p}_g^{-1}\Big(\E_n\Big[\indicator{G_i=g}(Y_{it} - Y_{ig-1})\Big] - \E_n\Big[A_i(g)\Big]'\widehat{\delta}^*(g,t)\Big) \right. \\
    &\hspace{35pt} \left. - p_g^{-1}\Big(\E\Big[\indicator{G=g}(Y_{t} - Y_{g-1})\Big] - \E\Big[A_i(g)\Big]'\delta^*(g,t)\Big) \right\}\\[5pt]
    &= \hat{p}_g^{-1}\sqrt{n} \left(\E_n\Big[\indicator{G_i=g}(Y_{it} - Y_{ig-1})\Big] - \E\Big[\indicator{G_i=g}(Y_{it} - Y_{ig-1})\Big] \right) \\
    &\hspace{15pt} - \hat{p}_g^{-1} \sqrt{n} \left( \E_n\Big[A_i(g)\Big]'\widehat{\delta}^*(g,t) - \E\Big[A_i(g)\Big]'\delta^*(g,t) \right) \\
    &\hspace{15pt} + \left(\E\Big[\indicator{G=g}(Y_{t} - Y_{g-1})\Big] - \E\Big[A_i(g)\Big]'\delta^*(g,t) \right)\sqrt{n}\Big( \hat{p}_g^{-1} - p_g^{-1}\Big) \\[5pt]
    &= \hat{p}_g^{-1}\sqrt{n} \left(\E_n\Big[\indicator{G_i=g}(Y_{it} - Y_{ig-1})\Big] - \E\Big[\indicator{G_i=g}(Y_{it} - Y_{ig-1})\Big] \right) \\
    &\hspace{15pt} - \hat{p}_g^{-1} \E_n\Big[A_i(g)\Big]'\sqrt{n}\Big(\widehat{\delta}^*(g,t) - \delta^*(g,t) \Big)\\
    &\hspace{15pt} - \hat{p}_g^{-1} \sqrt{n} \left( \E_n\Big[A_i(g)\Big] - \E\Big[A_i(g)\Big] \right)' \delta^*(g,t) \\
    &\hspace{15pt} - \hat{p}_g^{-1} ATT(g,t)\sqrt{n}\Big( \hat{p}_g - p_g\Big) \\[5pt]
    &= \frac{1}{\sqrt{n}} \sum_{i=1}^n \Big( \psi_{igt}^{(1)} + \psi_{igt}^{(2)} + \psi_{igt}^{(3)} + \psi_{igt}^{(4)} \Big) + o_p(1) \\[5pt]
    &= \frac{1}{\sqrt{n}} \sum_{i=1}^n \psi_{igt} + o_p(1) 
\end{align*} 
where the first equality holds from the expressions for $\widehat{ATT}(g,t)$ and $ATT(g,t)$ in \Cref{eqn:attgt-est,eqn:attgt-hat}, the second equality holds by adding and subtracting terms, the third equality adds and subtracts more terms (for the term from the second line of the previous equality) and cross-multiplies and uses the definition of $ATT(g,t)$ (for the term from the third line of the previous equality), and the fourth equality holds by the definitions of $\psi_{igt}(j)$ for $j=1,2,3,4$ and uses \Cref{Lemma_deltaInfluence} for $\psi_{igt}^{(2)}$. Part (ii) follows from \Cref{Lemma_deltaInfluence} and \Cref{ass:sel,ass:ife} (which together imply \Cref{eqn:Eu} whence $\E[\psi_{igt}^{(2)}]=0$) as the above implies $\widehat{ATT}(g,t)$ converges in probability to $ ATT(g,t)$ at the $\sqrt{n}$-rate. For $(g,t)$-fixed, observe that $\psi_{igt}$ are $iid$ across $i=1,\dots,n$ under \Cref{ass:sampling}. The conclusion of part (iii) follows from the Lindberg-L\'evy Central Limit Theorem (CLT).
\end{proof}

\bigskip

\begin{proof}[\textbf{Proof of \Cref{Corollary:ATT}}]
    The asymptotic linear representation simply follows from \Cref{Theorem:Asymp_Normal} by vertically stacking the elements of the set $\{\psi_{igt}\}_{(g,t)\in\mathcal{G}^\dagger \times \{g, \ldots, t^{max}(g)\}}$ into the vector $\Psi_i$. Observing that under \Cref{ass:sampling}, the summands $\Psi_i$ are $iid$, the second part follows from the multivariate Lindberg-L\'evy CLT. Lastly, consider the following decomposition.
    \begin{align*}
        \sqrt{n}(\hat{w}'\widehat{ATT} - w'ATT) &= \sqrt{n}w'(\widehat{ATT} - ATT) + \sqrt{n}ATT'(\hat{w}-w) + \sqrt{n}(\hat{w}-w)'(\widehat{ATT} - ATT)\\
        &= \frac{1}{\sqrt{n}}\sum_{i=1}^n\Big(w'\Psi_i + ATT'\mathcal{W}_i\Big) + o_p(1)
    \end{align*}
    \noindent The second equality follows from $\sqrt{n}$-consistency of $\hat{w}$. As the summands $w'\Psi_i + ATT'\mathcal{W}_i$ are $iid$ and mean-zero, the conclusion follows from the CLT.
\end{proof}

\section{Additional Details about Rank Conditions} \label{app:rank-explanation}

In this section, we show that the interactive fixed effects model in \Cref{ass:ife} reduces to one with fewer interactive fixed effects if either $\textrm{Rank}(\bmd \mathbf{F}) < R$ or $\textrm{Rank}(\mathbf{\Lambda}) < R+1$.
Recall that 
\begin{align*}
    \Delta Y_i(0) &= \Delta \theta + \bmd \mathbf{F} \lambda_i + \Delta e_i
\end{align*}

First, consider the case where $\textrm{Rank}(\mathbf{\Lambda}) = (R+1)$, but where $\textrm{Rank}(\bmd \mathbf{F}) = R-1$.  This means that one of the columns of $\bmd \mathbf{F}$ can be written as a linear combination of the other columns of $\bmd \mathbf{F}$.  Without loss of generality, suppose that it is the last column.  Partition $\bmd \mathbf{F} = \begin{bmatrix}\bmd \mathbf{F}_{,1:(R-1)} & \bmd \mathbf{F}_{,R} \end{bmatrix}$ where $\bmd \mathbf{F}_{,1:(R-1)}$ is a $(\mathcal{T}-1)\times R$ matrix that contains the first $(R-1)$ columns of $\bmd \mathbf{F}$ and $\bmd \mathbf{F}_{,R}$ is an $(\mathcal{T}-1)$ dimensional vector that contains the last column of $\bmd \mathbf{F}$. There exists an  $(R-1)$ dimensional vector $c$ such that $\bmd \mathbf{F}_{,R} = \bmd \mathbf{F}_{,1:(R-1)} c$.  Thus, we can re-write
\begin{align*}
    \Delta Y_i(0) &= \Delta \theta + \begin{bmatrix} \bmd \mathbf{F}_{,1:(R-1)} & \bmd \mathbf{F}_{,R} \end{bmatrix} \lambda_i + \Delta e_i \\
    &= \Delta \theta + \begin{bmatrix} \bmd \mathbf{F}_{,1:(R-1)} & \bmd \mathbf{F}_{,1:(R-1)} c \end{bmatrix} \lambda_i + \Delta e_i \\
    &= \Delta \theta + \bmd \mathbf{F}_{,1:(R-1)} \lambda_{i,1:(R-1)} + \bmd \mathbf{F}_{,1:(R-1)} c \lambda_{i,R} + \Delta e_i \\
    &= \Delta \theta + \bmd \mathbf{F}_{,1:(R-1)} \Big( \lambda_{i,1:(R-1)} + c \lambda_{i,R}\Big) + \Delta e_i
\end{align*}
This is an interactive fixed effects model with $(R-1)$ factors, as claimed.  Before moving on, it is worth providing a bit more intuition about the differences between cases where $\textrm{Rank}(\bmd \mathbf{F}) = (R-1)$ relative to the case considered in the main text where $\textrm{Rank}(\bmd \mathbf{F}) = R$.  Relative to the main case, the case considered here amounts to there being less independent variation in the factors across time (to be clear, we are maintaining that $(\mathcal{T}-1) \geq R$ so that there are enough available periods).  What we have shown is that, in this case, the model reduces to one in which there are fewer factors (as we mentioned in the main text).

Now, consider the case where $\textrm{Rank}(\bmd \mathbf{F}) = R$, but where $\textrm{Rank}(\mathbf{\Lambda}) = R$.  Notice that we can write
\begin{align*}
    \Delta Y_{i}(0) = \Delta \theta + \bmd \mathbf{F} \E[\lambda | G] + (\Delta e_i + \bmd \mathbf{F} u_i)
\end{align*}
where $u_{i} = \lambda_i - \E[\lambda | G]$ which is an $R$ dimensional vector that is mean independent of $G$.  In this case, one of the columns of $\mathbf{\Lambda}$ can be written as a linear combination of the other columns of $\mathbf{\Lambda}$.  Without loss of generality, suppose that the last column of $\mathbf{\Lambda}$, and partition $\mathbf{\Lambda} = \begin{bmatrix}\mathbf{\Lambda}_{,1:(R-1)} & \mathbf{\Lambda}_{,R}\end{bmatrix}$ where
\begin{align*}
    \mathbf{\Lambda}_{,1:(R-1)} = \begin{bmatrix} 1 & \E[\lambda_1 | G=g'] & \cdots & \E[\lambda_{R-1} | G=g'] \end{bmatrix}_{g' \in \mathcal{G}} \qquad \mathbf{\Lambda}_{,R} = \begin{bmatrix} \E[\lambda_R | G=g']\end{bmatrix}_{g' \in \mathcal{G}}    
\end{align*}
which are a $|\mathcal{G}| \times R$ matrix and $|\mathcal{G}|$ dimensional vector respectively.  In the case considered here, there exists an $R$ dimensional vector $c$ such that we can write $\mathbf{\Lambda}_{,R} = \mathbf{\Lambda}_{,1:(R-1)} c$.  To make progress along these lines, we need to introduce some more notation.  Define the elements of $c$ according to $c=(c_0, c_1, \ldots, c_{R-1})'$.  Also define $c_{\neg 0} = c \setminus c_0$ (i.e., all of the elements of $c$ except for $c_0$).  Similarly, define $u_{i\neg R} = u_i \setminus u_{iR}$ (i.e., all of the elements of $u_i$ except the last one).

Then, we have that
\begin{align*}
    \Delta Y_i(0) &= \Delta \theta + \begin{bmatrix} \bmd \mathbf{F}_{,1:(R-1)} & \bmd \mathbf{F}_{,R} \end{bmatrix}\begin{pmatrix} \E[\lambda_1 | G] \\ \vdots \\ \E[\lambda_{R-1} | G] \\ \E[\lambda_R | G] \end{pmatrix} + \left(\Delta e_i + \begin{bmatrix} \bmd \mathbf{F}_{,1:(R-1)} & \bmd \mathbf{F}_{,R}\end{bmatrix} u_i\right) \\
    &= \Delta \theta + \begin{bmatrix} \bmd \mathbf{F}_{,1:(R-1)} & \bmd \mathbf{F}_{,R} \end{bmatrix}\begin{pmatrix} \E[\lambda_1 | G] \\ \vdots \\ \E[\lambda_{R-1} | G] \\ c_0 + c_1 \E[\lambda_1 | G] + \cdots c_{R-1} \E[\lambda_{R-1} | G] \end{pmatrix} + \left(\Delta e_i + \begin{bmatrix} \bmd \mathbf{F}_{,1:(R-1)} & \bmd \mathbf{F}_{,R}\end{bmatrix} u_i\right) \\
    &= \Delta \theta + \begin{bmatrix} \bmd \mathbf{F}_{,1:(R-1)} & \bmd \mathbf{F}_{,R}\end{bmatrix} \begin{pmatrix} \lambda_{i1} - u_{i1} \\ \vdots \\ \lambda_{iR-1} - u_{iR-1} \\ c_0 + c_1 (\lambda_{i1} - u_{i1}) + \cdots + c_{R-1}( \lambda_{iR-1} - u_{iR-1}) \end{pmatrix} \\
    & \hspace{10pt} + \left(\Delta e_i + \begin{bmatrix} \bmd \mathbf{F}_{,1:(R-1)} & \bmd \mathbf{F}_{,R}\end{bmatrix} u_i\right) \\ 
    &= \Delta \theta + \bmd \mathbf{F}_{,R} c_0 + \bmd \mathbf{F}_{,1:(R-1)} \begin{pmatrix} \lambda_1 \\ \vdots \\ \lambda_{R-1} \end{pmatrix} + \bmd \mathbf{F}_{,R} c_1 \lambda_1 + \cdots + \bmd \mathbf{F}_{,R} c_{R-1} \lambda_{R-1} \\
    & \hspace{10pt} - \bmd \mathbf{F}_{,1:(R-1)} u_{i\neg R} - \bmd \mathbf{F}_{,R} c_1 u_{i1} - \cdots - \bmd \mathbf{F}_{,R} c_{R-1} u_{iR-1} 
    + \left(\Delta e_i + \begin{bmatrix} \bmd \mathbf{F}_{,1:(R-1)} & \bmd \mathbf{F}_{,R}\end{bmatrix} \begin{pmatrix} u_{i\neg R} \\ u_{iR} \end{pmatrix}\right) \\
    &= \underbrace{\Big(\Delta \theta + \bmd \mathbf{F}_{,R} c_0\Big)}_{\small \textrm{time fixed effects}} + \underbrace{\begin{bmatrix} \bmd \mathbf{F}_{,1:(R-1)} + \bmd \mathbf{F}_{,R} c_{\neg 0}' \end{bmatrix}}_{\small \textrm{factors}} \begin{pmatrix} \lambda_1 \\ \vdots \\ \lambda_{R-1} \end{pmatrix} + \underbrace{\left\{ \Delta e_i + \bmd \mathbf{F}_{,R} \Big( u_{iR} -  c_{\neg 0}' u_{i\neg R}\Big) \right\}}_{\small \textrm{idiosyncratic error}}
\end{align*}
where the first equality holds from the interactive fixed effects model in \Cref{ass:ife}, by the definition of $u_i$ above, and by partitioning $\bmd \mathbf{F}$ in the same way as above; the second equality holds given the reduced rank of $\mathbf{\Lambda}$ discussed above; the third equality holds by the definition of $u_i$; the fourth equality holds by combining terms; the fifth equality also holds by combining terms, re-arranging, and canceling.  In the last line, the idiosyncratic error term is mean independent of the groups.  This discussion implies that, when $\textrm{Rank}(\mathbf{\Lambda}) = R$ (rather than $R+1$), the interactive fixed effects model reduces from a model with $R$ interactive fixed effects to one with $(R-1)$ interactive fixed effects.  

The above discussion is rather technical, and it is worth discussing what reduced rank of $\mathbf{\Lambda}$ means in a particular example.  Consider the baseline case discussed in \Cref{SubSect:Baseline} except suppose that $R=2$.  In this case, 
\begin{align*}
    \mathbf{\Lambda} = \begin{bmatrix} 1 & \E[\lambda_1 | G=3] & \E[\lambda_2 | G=3] \\
    1 & \E[\lambda_1 | G=4] & \E[\lambda_2 | G=4] \\
    1 & \E[\lambda_1 | G=\infty] & \E[\lambda_2 | G=\infty] \end{bmatrix}
\end{align*}
The rank would be reduced if it were the case that $\E[\lambda_1 | G=g] = \kappa_1 \E[\lambda_2 | G=g]$ for all groups, for some constant $\kappa_1$.  Inasmuch as $\lambda_1$ and $\lambda_2$ are unobserved heterogeneity, this is a rather strange case.  The rank can also be reduced if, for example, 
\begin{align*}
    \begin{bmatrix} \E[\lambda_1 | G=3] \\ \E[\lambda_2 | G=3] \end{bmatrix} =  \begin{bmatrix} \E[\lambda_1 | G=4] \\ \E[\lambda_2 | G=4] \end{bmatrix}    
\end{align*} 
i.e., the mean of $\lambda_1$ and $\lambda_2$ is the same for groups 3 and 4.  This is a realistic possibility in applications. $\E[\lambda_1|G=\infty]$ and $\E[\lambda_2|G=\infty]$ are unrestricted, but we can write $\E[\lambda_1|G=\infty] = \kappa_2 \E[\lambda_2|G=\infty]$ for some unknown constant $\kappa_2$.  It is useful for the argument below to note that  $(\E[\lambda_1|G] - \E[\lambda_1|G=3]) = \kappa_2(\E[\lambda_2|G] - \E[\lambda_1|G=3)$.  This holds because, for $G=3$ or $G=4$, both sides are equal to 0.  For $G=\infty$ the equality holds because $\E[\lambda_1|G=\infty] = \kappa_2 \E[\lambda_2|G=\infty]$.

In this case, following a simplified version of the argument presented above, we have that
\begin{align*}
    \Delta Y_i(0) &= \Delta \theta + \bmd \mathbf{F}_1 \lambda_{i1} + \bmd \mathbf{F}_2 \lambda_{i2} + \Delta e_i \\
    &= \Delta \theta + \bmd \mathbf{F}_1 \E[\lambda_1 | G] + \bmd \mathbf{F}_2 \E[\lambda_2 | G] + \Big( \Delta e_i + \bmd \mathbf{F}_1 u_{i1} + \bmd \mathbf{F}_2 u_{i2} \Big) \\
    &= \Delta \theta + \bmd \mathbf{F}_1 \E[\lambda_1|G=3] + \bmd \mathbf{F}_2 \E[\lambda_2|G=3] \\
    &\hspace{10pt} + \bmd \mathbf{F}_1 \Big(\E[\lambda_1 | G]-\E[\lambda_1|G=3]\Big) + \bmd \mathbf{F}_2 \Big(\E[\lambda_2 | G] - \E[\lambda_2 | G=3]\Big) \\
    &\hspace{10pt} + \Big( \Delta e_i + \bmd \mathbf{F}_1 u_{i1} + \bmd \mathbf{F}_2 u_{i2} \Big) \\
    &= \Delta \theta + \bmd \mathbf{F}_1 \E[\lambda_1|G=3] + \bmd \mathbf{F}_2 \E[\lambda_2|G=3] \\
    &\hspace{10pt} + \bmd \mathbf{F}_1 \kappa_2 \Big(\E[\lambda_2 | G]-\E[\lambda_2|G=3]\Big) + \bmd \mathbf{F}_2 \Big(\E[\lambda_2 | G] - \E[\lambda_2 | G=3]\Big) \\
    &\hspace{10pt} + \Big( \Delta e_i + \bmd \mathbf{F}_1 u_{i1} + \bmd \mathbf{F}_2 u_{i2} \Big) \\
    &= \Delta \theta + \bmd \mathbf{F}_1 \E[\lambda_1|G=3] - \bmd \mathbf{F}_1 \kappa_2 \E[\lambda_2 | G=3]\\
    &\hspace{10pt} + \Big(\bmd \mathbf{F}_1 \kappa_2 + \bmd \mathbf{F}_2 \Big) \E[\lambda_2 | G]  \\
    &\hspace{10pt} + \Big( \Delta e_i + \bmd \mathbf{F}_1 u_{i1} + \bmd \mathbf{F}_2 u_{i2} \Big) \\
    &= \underbrace{\Delta \theta + \bmd \mathbf{F}_1 \E[\lambda_1|G=3] - \bmd \mathbf{F}_1 \kappa_2 \E[\lambda_2 | G=3]}_{\small \textrm{time fixed effects}} \\
    &\hspace{10pt} + \underbrace{\Big(\bmd \mathbf{F}_1 \kappa_2 + \bmd \mathbf{F}_2 \Big)}_{\small \textrm{factor}} \lambda_{i2}  \\
    &\hspace{10pt} + \underbrace{\Big( \Delta e_i + \bmd \mathbf{F}_1 u_{i1} - \bmd \mathbf{F}_{1} \kappa_2 u_{i2} \Big)}_{\small \textrm{idiosyncratic error term}}
\end{align*}
where the first equality comes from the definition of the interactive fixed effects model with two interactive fixed effects, the second equality uses the definition of $u_{ij}$ for $j=1,2$, the third equality adds and subtracts $\bmd \mathbf{F}_j \E[\lambda_j|G=3]$ for $j=1,2$, the fourth equality replaces $(\E[\lambda_1|G] - \E[\lambda_1|G=3])$ using the argument in the preceding paragraph, the fifth equality re-arranges and cancels terms, and the last equality uses the definition of $u_{i2}$ again to substitute for $\E[\lambda_2|G]$ and then re-arranges and cancels terms.  This shows that, in the case considered here, the interactive fixed effects model reduces to one with only one interactive fixed effect.

\end{document}